%% file: main.tex
\documentclass[runningheads]{llncs}
\usepackage{geometry} 
\usepackage[utf8]{inputenc}
\usepackage{booktabs} 
\usepackage[ruled]{algorithm2e} 
\usepackage[ruled]{algorithm2e} 

\SetAlFnt{\small}
\SetAlCapFnt{\small}
\SetAlCapNameFnt{\small}
\SetAlCapHSkip{0pt}
\IncMargin{-\parindent}

\usepackage{comment}
\usepackage{xcolor}
\usepackage[square,sort,comma,numbers]{natbib}
\usepackage{hyperref}
\usepackage{amsmath,amssymb,commath,dsfont}
\usepackage{graphicx}

\newcommand\cL[1]{c_{#1}^{\textnormal{life}}} 
\newcommand\cP[1]{c_{#1}^{\textnormal{platform}}} 
\newcommand\ceff[1]{\cP{#1} + \frac{\lambda_{#1}}{z_{#1}}\left(\cP{#1} - \cL{#1}\right)} 
\newcommand\ceffJ[2]{\cP{#1} + \frac{\lambda_{#1}}{z_{#1}(#2)}\left(\cP{#1} - \cL{#1}\right)}
\newcommand\eps{\epsilon}
\newcommand\argmax[1]{\underset{#1}{\textnormal{argmax }}}

\setcitestyle{acmnumeric}

\begin{document}

\title{The Platform Design Problem}
%
%
\author{Christos Papadimitrio$\textnormal{u}^*$\inst{1} \and
Kiran Vodrahall$\textnormal{i}^*$\inst{2} \and
Mihalis Yannakaki$\textnormal{s}^*$\inst{3}}
\authorrunning{Papadimitriou, Vodrahalli, and Yannakakis}
%
\institute{Department of Computer Science, Columbia University, New York, NY; 
\email{christos@columbia.edu}\\
\and
Department of Computer Science, Columbia University, New York, NY; 
\email{kiran.vodrahalli@columbia.edu}\\
\and
Department of Computer Science, Columbia University, New York, NY; 
\email{mihalis@cs.columbia.edu}}
\maketitle              

\begin{abstract}
On-line firms deploy suites of software platforms, where each platform is designed to interact with users during a certain activity, such as browsing, chatting, socializing, emailing, driving, etc.  The economic and incentive structure of this exchange, as well as its algorithmic nature, have not been explored to our knowledge.  We model this interaction as a Stackelberg game between a Designer and one or more Agents.  We model an Agent as a Markov chain whose states are activities; we assume that the Agent's utility is a linear function of the steady-state distribution of this chain. The Designer may design a platform for each of these activities/states; if a platform is adopted by the Agent, the transition probabilities of the Markov chain are affected, and so is the objective of the Agent.  The Designer's utility is a linear function of the steady state probabilities of the accessible states (that is, the ones for which the platform has been adopted), minus the development cost of the platforms. The underlying optimization problem of the Agent --- that is, how to choose the states for which to adopt the platform --- is an MDP. If this MDP has a simple yet plausible structure (the transition probabilities from one state to another only depend on the target state and the recurrent probability of the current state) the Agent's problem can be solved by a greedy algorithm.  The Designer's optimization problem (designing a custom suite for the Agent so as to optimize, through the Agent's optimum reaction, the Designer's revenue), is in general NP-hard to approximate within any finite ratio; however, in the special case, while still NP-hard, has an FPTAS.  These results generalize, under mild additional assumptions, from a single Agent to a distribution of Agents with finite support, as well as to the setting where other Designers have already created platforms, and the Designer must find the best response to the strategies of the other Designers. We discuss other implications of our results and directions of future research. 
\end{abstract}

\input{intro.tex}
\input{agent.tex}
\input{designer.tex}

\input{manyagents.tex}

\input{competitive.tex}
\input{discussion.tex}

\input{acknowledgements.tex} 
\bibliographystyle{plainnat}
\bibliography{main}

\newpage 
\appendix 

\input{appendixA.tex} 
\input{appendixB.tex} 
\input{appendixC.tex}
\input{appendixD.tex}

\input{appendixE.tex}
\input{appendixF.tex}

\end{document}

%% file: intro.tex
\section{Introduction} \label{sec:intro}
In  economics, the creation of wealth happens through markets: environments in which firms employ land, labor, capital, raw materials, and technology to produce new goods for sale, at equilibrium prices, to consumers and other firms. Since all agents in this scenario participate voluntarily, wealth must be created.  Accordingly, markets have been the focus of a tremendous intellectual effort by economists, mathematicians, and, more recently, computer scientists. 

Over the past three decades the global information environment has spawned novel business models seemingly beyond the reach of the extant theory of markets, and which, arguably, account for a large part of present-time wealth creation, chief among them a new kind of software company that can be called {\em platform designer}.  On-line platforms are created with which consumers interact during certain activities: search engines facilitate browsing, social networks host social interactions, movie, music, and game sites provide entertainment, chatting and email apps mediate communication. Shopping platforms, navigation maps, tax preparation sites, and many more platforms bring convenience and therefore value to consumers' lives.  Increasingly during these past two decades, on-line firms have created comprehensive {\em suites} of platforms, covering many such life activities. Platform designers draw much of their revenue through the data that they collect about the users interacting with their platforms, which data they either sell to other firms or use to further fine tune and enhance their own business.  In this paper we point out that, in the case of platform designers, the most elementary aspects of markets, for example the theory of production and consumption, are quite nontrivial. We focus on a restricted case of the problem corresponding to the ``substitutes'' case, having proved that the case with complements (when platforms are allowed to feed into one another) is hopeless.  Note that this reflects the history of the search, in the market context, for conceptually, and implicitly computationally, tractable cases. (Recall the fruitful early work by Arrow and other economists on the identification of classes of markets with good structural properties, such as the gross substitutes case \citep{Arrow59}, and the extensive more recent work in computer science developing algorithms for special cases, like the case of linear utilities \citep{Vazirani07}.)

\paragraph{Our Model and Results} We model the platform design problem as a Stackelberg game (that is, a game where one player goes first and the others react optimally) with two players, a Designer and an Agent (the extension to many Agents is also studied, and the case of many competing Designers is also discussed). Here, the Designer plays first, and the Agent responds.
The Agent is modeled as an ergodic Markov chain on a set of states $\mathcal{A}$, representing the Agent's life activities. 
We assume that the Agent receives a fixed payoff per unit of time spent at each state. 
The Designer has the opportunity to design a platform for each state in $\mathcal{A}$, which the Agent may or may not choose to adopt. There is a one-time cost for the Designer to build a platform for a given state. If the Agent adopts the platform, the transitions of the Agent's life change at that state, and the Agent's utility at that state may increase or decrease as a result of adoption\footnote{One possible reason for diminished utility is the aversion of the Agent to the Designer's access to personal information pertaining to that state.}. In return, the Designer gets to observe the Agent at that state and derives a fixed utility payoff for the fraction of time the Agent spends in that state. We assume that platform revenue is proportional to the time users spend on the platform, which strikes us as a reasonable first approximation.

We note immediately that the Agent's optimization problem, once the Designer has deployed a set $S$ of platforms, is a Markov decision process (MDP), and it follows from MDP theory that the Agent will adopt some of the platforms offered and reject the rest 
and the optimum set of adopted platforms can be computed by linear programming (and other methods).

Now the {\em platform design problem} (PDP) is the following:  Given the Markov chain, all utility coefficients for both the Agent and the Designer, and the development costs of the platforms, choose a set of states $S$ for which to create platforms, so as to maximize the Designer's utility; namely, the utility to the Designer of the Markov chain that results from the optimum response by the Agent to the platforms in $S$, minus the development costs of the platforms in $S$. It is immediate that, since the Designer can anticipate Agent's optimal response, at optimality all platforms in the optimum set $S$ will be adopted.

We show that PDP is NP-hard to approximate within {\em any} finite ratio (Theorem \ref{thm:inapprox}).  The proof of this result is quite instructive, because it relies almost exclusively on the fact that introduced platforms can modify the Markov chain so as to funnel traffic from one platform to the other, and therefore create the stark choices necessary for this level of complexity. The construction has the property that offering a platform in one state can make it more attractive for the Designer to offer a platform also at another state (if the adoption of the platform in the first state increases the transition probability to the second state)  In economic terms, the platforms offered by the Designer can be {\em complementary goods}, and making decisions for such goods tend to be difficult.  

In view of this obstacle, we next turn to a special kind of Markov chain, for which platforms are essentially {\em substitute goods;} generally, substitution is known to lead to better behaving markets.  A Markov chain of this sort, called {\em the flower} (see Figure~\ref{fig:life-platform-chain}), has a number of transition parameters that is linear in $|\mathcal{A}|$.  At each state $i$, the transition leads back to the state with some probability $q_i$, while the rest of the probability $(1-q_i)$ is split among the other states {\em in proportions that are fixed.}  Evidently, this is equivalent to a chain that has an extra ``rest state'' $0$ with $q_0=0$, that is, a purely transitional state (see Figure~\ref{fig:life-platform-chain}).
Adopting a platform now increases or decreases the transition probability of the state to itself, decreasing or increasing, respectively, the transition probability to the other states. We show that, in this case, the MDP optimizing the Agent's objective, given the available platforms, becomes a quasiconcave combinatorial optimization problem with special structure (Lemma~\ref{thm:equiv-ip}), which can be solved by a greedy algorithm (Theorem~\ref{thm:greedy-is-opt}).
The algorithm can be extended to a setting where there are multiple available platforms for each state in $\mathcal{A}$, and the agent can choose to adopt one or none of these options for each state (Theorem~\ref{thm:agent-multi-opt}).

The PDP in the flower specal case is still NP-hard (Theorem~\ref{thm:hardness}), but has a dynamic programming FPTAS if one parameter --- the expected time spent at each state --- is quantized (Theorem~\ref{thm:FPTAS-pdp}). 
The dynamic programming algorithm can be extended, through some further quantization, to the case of many agents --- except that the number of agents is now in the exponent (Theorem~\ref{thm:multi-Agent-PDP}).  Given that the number of agents is likely to be very large, the best way to think of this algorithm is as an algorithm for the case in which one is given a {\em distribution} of agents of finite support --- that is, with a small number of {\em agent types.}
Similarly, essentially the same algorithm can be adapted to the competitive setting, where a Designer enters a field where many Designers have already built existing platforms, and must now decide which platforms to build (Theorem~\ref{thm:multi-Designer-PDP}).


\paragraph{Related Work}
We are not aware of past research on the production and consumption of online platforms. Computational aspects of Stackelberg games between consumers and firms designing or packaging on-line products have been explored to a small degree, see e.g.~\citep{Kleinberg98a,Kleinberg98b}. There has been work on online decision making, where at each round the Designer gets to select from some set of options (e.g., which is the best ad to display to the user of a website) and receives a reward after deployment for that round, as well as additional information about the performance of the other options \citep{CesaBianchi06}; see also~\citep{Frazier14, Mansour15, Roughgarden16, Liu18, Lykouris19}. This line of research is of obvious relevance to the present one, even though our Agent model is far more complex.  More recently, trade-offs in on-line activity by consumers, for example between effectiveness of browsing and privacy, have been discussed \citep{Tsitsiklis18a,Tsitsiklis18b}.  The ways in which on-line firms profit from data has been somewhat explored, see e.g.~\citep{Agarwal19} but not in any manner that can be used in our model; here we consider it a given parameter.

\paragraph{Our Contributions} 
Our main contributions are: the articulation of the Platform Design Problem, the observation that it is profoundly intractable in its generality, the identification of the tractable class of flower Markov chains, roughly corresponding to substitution in markets, the solution of the Agent's and the Designer's problems through the Agent's greedy algorithm and dynamic programming, 
the generalizations of these algorithms to multiple Agents and Designers,
and the many directions for further research opened (see the discussion in Section~\ref{sec:discussion}).  

\section{PDP: Intractability of the General Case}
\label{sec.general_pdp}
\input{description_of_general_problem}

We prove that the PDP in its generality is as severely intractable as any optimization problem can be:  It is NP-hard to approximate {\em within any finite approximation ratio.}

\begin{theorem}\label{thm:inapprox}
It is strongly NP-hard to decide whether the optimum solution to a PDP instance has zero or positive profit for the designer.
\end{theorem}
\begin{proof}
We reduce from the Set Cover problem. 
Given a family $F$ of $m$ subsets of a set $U$ of $n$ elements and an integer $k$, we want to determine if there is a subfamily of $F$ with $k$ sets whose union is $U$.
We define an instance of the PDP problem as follows. There are $m+n+1$ states, one for each set of $F$ and each element of $U$, and an additional `bad' state. 
For each set-state $S_i$, there is one potential platform $p(S_i)$ that the Designer may decide to offer at the state $S_i$. For each element state $u_j$ and every set $S_i$ of $F$ that contains element $u_j$ there is a platform $p(u_j,S_i)$ that the Designer may offer at state $u_j$; the Designer will offer at most one
of these platforms at state $u_j$.\footnote{We allow here the Designer to have a choice among several platforms in a state; it is easy to modify the construction, by using additional states, so that in each state the Designer has only one potential platform, which she may choose to build.} The Designer has no platform for the last `bad' state. 

The Agent likes all the platforms: that is, the Agent's rewards are such that he will adopt every platform that is offered by the Designer. Initially the MDP is at any element-state $u_j$ with uniform probability $1/n$. The transition probabilities of the Agent's MDP are as follows. An element-state $u_j$ with platform $p(u_j,S_i)$ (if adopted) transitions with probability $1$ to the set-state $S_i$.
An element state $u_j$ with no adopted platform transitions with probability $1$ to the bad state.
A set-state $S_i$ with adopted platform $p(S_i)$ self-loops with probability $1-1/k^2$ and transitions with the remaining probability to a uniformly random element-state.
A set-state $S_i$ with no (adopted) platform transitions with probability $1$ to the bad state.
The bad state self-loops with probability $1-1/nk^4$ and transitions with the remaining probability to a uniformly random element-state.

The Designer's rewards and costs are as follows. The reward rate for each set-state platform $p(S_i)$ is set to $r=k^2+k$, i.e. the Designer receives revenue equal to $r$ times the fraction of the time that the Agent spends in platform $p(S_i)$; the cost of building the platform is $k$.
The reward rates and costs of the platforms $p(u_j,S_i)$ are set to 0.
The objective of the Designer is to select a set of platforms to offer that maximizes the total profit, which is the total reward minus the total cost.

We claim that the optimal profit for the Designer is positive if and only if the Set Cover instance has a solution with at most $k$ sets. Intuitively, the goal of the Designer is to keep the Agent at all times within her "ecosystem", i.e. in states with her platforms, while making a profit.

First, suppose that there is a set cover $C$ with at most $k$ sets. The Designer offers the platform $p(S_i)$ for every $S_i \in C$ at the set-state $S_i$, and for each element-state $u_j$, the Designer offers a platform $p(u_j,S_i)$ for some $S_i \in C$ that contains $u_j$. The cost of building the platforms is $k|C| \leq k^2$. The Agent adopts all the offered platforms, and because of the transition probabilities, spends almost all the time at the set-states corresponding to sets in $C$, specifically a fraction $\frac{k^2}{1+k^2}$ of the time. Therefore, the profit of the Designer is at least $r\frac{k^2}{1+k^2} - k^2 >0$.

Conversely, suppose that the Designer has a solution with positive profit. Suppose that some element-state $u_j$ does not have a platform, or $u_j$ has a platform $p(u_j,S_i)$ but the corresponding set-state $S_i$ does not have the corresponding platform $p(S_i)$. Then, every time the MDP visits $u_j$ will
then move subsequently to the bad state. Therefore, the MDP will spend most of the time (specifically at least $1-1/k^2$ fraction of the time) in the bad state, which does not provide any revenue to the Designer. Thus, the total revenue to the Designer is at most $r/k^2$ which is less than the cost of a set platform.
We conclude that, if the profit is positive, then every element state $u_j$ must have a platform $p(u_j,S_i)$ and the corresponding state $S_i$ must have the corresponding platform $p(S_i)$. This implies that the collection $C$ of set-states $S_i$ with a platform forms a set cover.
The Designer's profit is at most $r-k|C|=k^2+k-k|C|$.
Since the profit is positive,  $|C| \leq k$.
\end{proof}

%% file: description_of_general_problem.tex

The platform design problem (PDP) is a Stackelberg game between a Designer and an Agent\footnote{We later consider the case with multiple Agents and multiple Designers, as well as multiple platforms per state.}. The Agent inhabits a discrete state space with transitions and rewards. The Designer moves first by building, at some fixed cost and for certain states, one platform per state. Each platform, if adopted by the Agent, changes the Agent's transitions and rewards at that state, and also yields to the Designer a reward rate (modeling the Designer's utility from learning about the Agent) per unit of time the Agent spends in the platform for each platform the Agent accepts. 
The Agent adopts platforms to optimize its expected reward in the resulting Markov Decision Process (MDP).  The Designer's goal is to build platforms so that the Agent behaves in a way that optimizes the Designer's total reward. Formally:

\begin{definition}[PDP]\label{def:PDP}
The Agent's environment is an irreducible Markov chain with state space $\mathcal{A}=[n]$ with $n$ states. At each state $i$, the transition probabilities out of $i$ are a vector of probabilities $T_i^{\textnormal{life}}$ and the reward coefficient is a real number $c_i^{\textnormal{life}}$.

The Designer chooses a set $S\subseteq [n]$ of these states for which to build {\em platforms}.  The Designer pays a fixed $\textnormal{cost}_{i}>0$ to build a platform at state $i$, and receives reward rate $d_{i}$ per unit of time the Agent spends at state $i$,  provided the Agent opts in to the platform at state $i$. 

After the Designer's move, the Agent faces a Markov Decision Problem (see \cite{Puterman94} for an introduction to Markov decision theory).  At each state $i\in S$, adoption of the platform will result in the transition probabilities changing to $T_i^{\textnormal{platform}}$ and the reward coefficient changing to $c_i^{\textnormal{platform}}$.  We assume that these changes in the transition probabilities are such that the reachable part of the Markov chain is irreducible\footnote{Irreducibility can be guaranteed by maintaining a cycle of tiny probability around the states; it will never be a problem in our arguments and constructions.}.  

The Agent's {\em optimal decision} in response to the Designer's move $S$ is a set $S'\subseteq S$ of states on which to adopt the platform (recall that in MDPs, it is well known that we can restrict the possible policies, without loss of optimality, to deterministic, Markovian, stationary policies computed by linear programming).  Let $M(S')$ be the Markov chain resulting from adopting the subset $S'$ of the platforms offered by the Designer. 

Coming now back to the Designer's first move, and since the Designer can fully anticipate the Agent's response $S'$ to $S$ and every extra platform has a positive cost, the Designer omits any platforms that would not be adopted --- that is, makes sure that $S=S'$. Among all such sets, the Designer chooses the one that optimizes
the Designer's {\em profit}
\[
\textnormal{profit}(S) :=  \sum_{i \in S} d_{i}\cdot \pi_i(S) - \sum_{i \in S} \textnormal{cost}_{i}
\]
$\pi_i(S)$ denotes the steady state distribution at state $i$ of the Markov $M(S)$. 
\end{definition}

%% file: agent.tex
\section{Flower Case: The Agent's Problem} \label{sec:agent}

The intractability proof of the general PDP in the previous section relies on the complementary nature of the construction: offering a platform in one state can make it more attractive for the designer to also offer a platform in certain other states. We will next define a special case of the PDP which is much better behaved, and in economic terms roughly corresponds to substitution.

An agent divides her time among the different states.
If the designer offers a platform at a state $s$ and the agent adopts it, she spends more time at $s$, and hence
has less time to spend in the rest of the states.
In the absence of complementarity, this means that it is now less beneficial for the designer to offer a platform in another state. In other words, platforms at different states compete for the attention (and the time) of the agent, and it is the agent's time spent on the platforms that determines their contribution to the profit of the designer.

We define now formally the model in the special case,
which will be be our focus in the rest of the paper.  





\begin{definition}[Flower MDP]\label{def:agent-mdp} 
We have the same setup as defined in Section~\ref{sec.general_pdp}, with some added constraints on what the possible transitions can be. We also add in a dummy state $0$ with no reward or platforms\footnote{Note that the rest state $0$ is for convenience and is not necessary in our model. We could equivalently have a graph where each node $i$ transitions to node $j$ with probability $(1 - q_i - y_i)\cdot p_j$, and self-transitions with probability $q_i + y_i + p_i(1 - q_i - y_i)$.}.
In Figure~\ref{fig:life-platform-chain}, we define the transitions of the Markov chain $T_{\textnormal{life}}$ to represent the Agent's life, and $T_{\textnormal{platform}}$ to represent the Agent's life when the platform is adopted at all states. Here, $p_i, q_i, y_i$ satisfy $\sum_i p_i = 1$, $0 < p_i$, $0 < q_i < 1$, and $0 < y_i < 1 - q_i$ for all $i \in [n]$. In words, $p_i$ denote transition probabilities to different states from the rest state, $q_i$ denote the self-transition probabilities, and $y_i$ denote the modification to the self-transitions due to the Agent accepting the platform at state $i$. 
\begin{figure}
    \centering
    \includegraphics[scale=0.45]{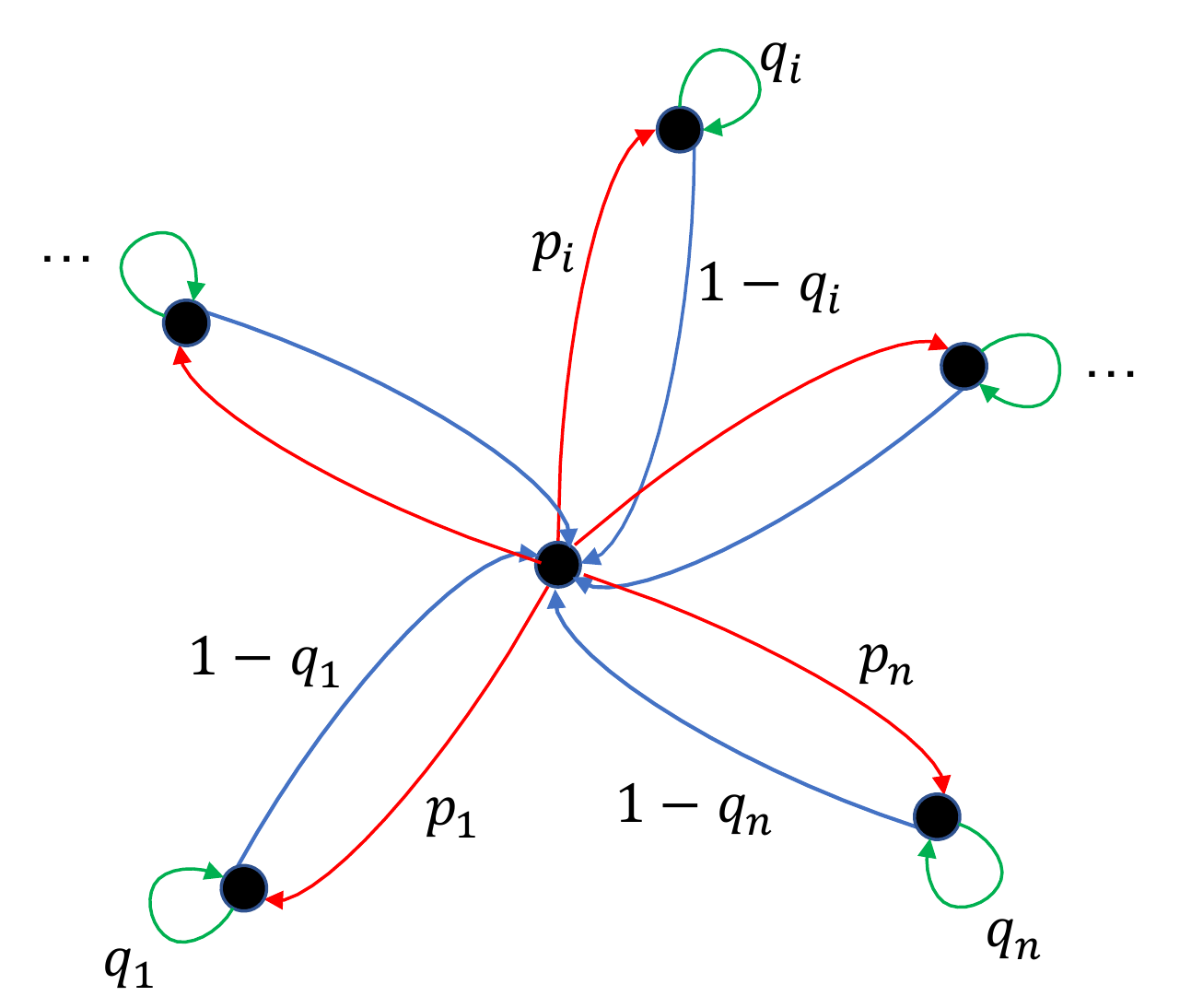}
    \includegraphics[scale=0.45]{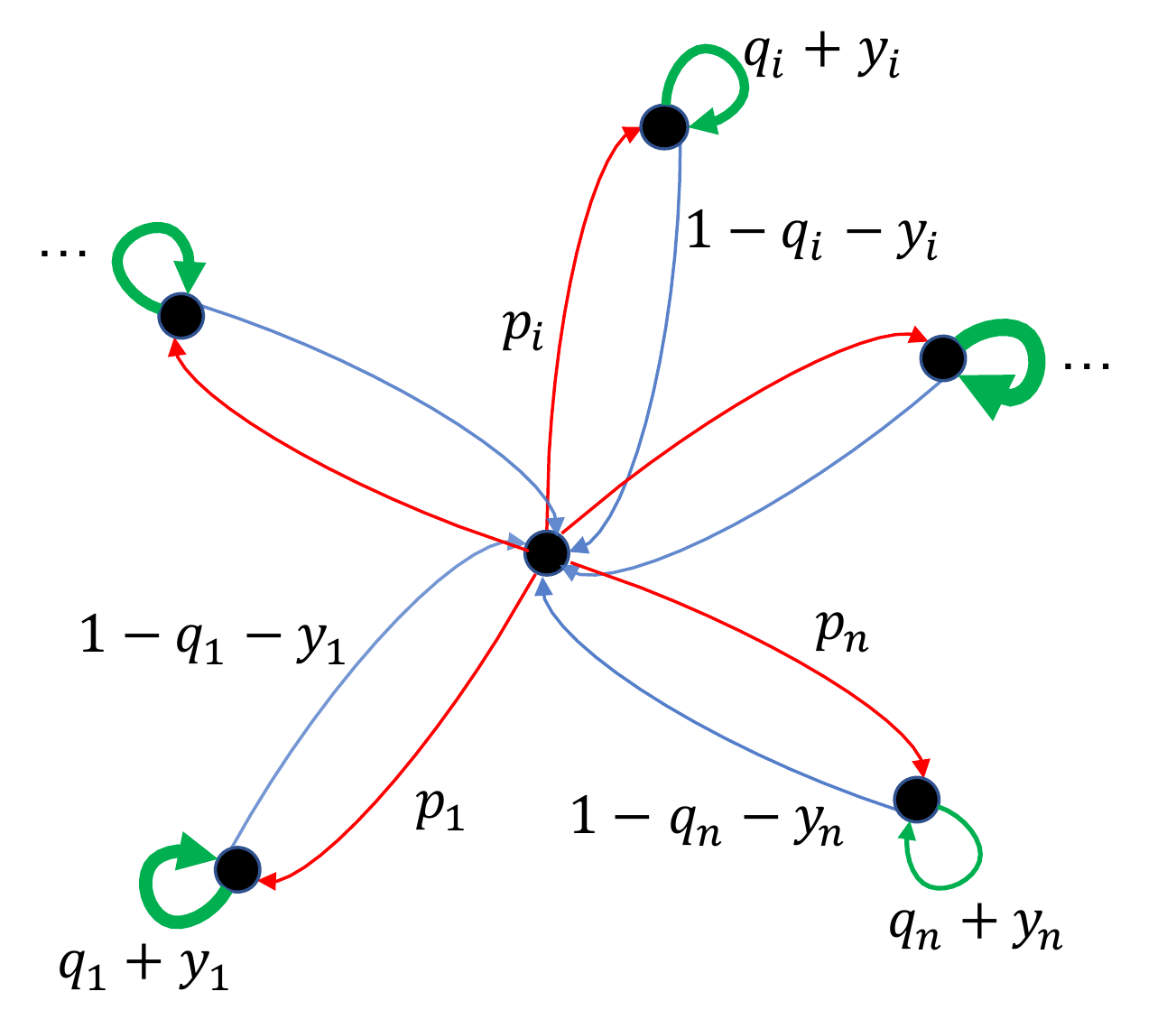}
    \caption{$T_{\textnormal{life}}$(left) and $T_{\textnormal{platform}}$(right).}
    \label{fig:life-platform-chain}
\end{figure}
At state $0$, the action chosen by the Agent does not affect the transitions, since the Designer never builds a platform there. 

\end{definition}


\subsection{The Greedy Algorithm}
Irreducible average-reward MDPs are efficiently solvable via linear programming, value and policy iteration, etc. \citep{Bertsekas17}. 
Here, we reformulate the Agent's problem as a combinatorial optimization problem with special structure, and solve it through a greedy algorithm.  The following is straightforward:

\begin{lemma} \label{thm:equiv-ip} The agent's objective for an optimal policy defined in Section~\ref{sec.general_pdp} can be re-written as the following optimization in the special case of the flower MDP (Definition~\ref{def:agent-mdp}):
\begin{align}
\begin{split}
\argmax{S \subseteq [n]} \frac{ A + \sum_{j \in S} z_j\phi(j)}{B + \sum_{j \in S} z_j} \label{eq:agent-opt}
\end{split}
\end{align}
where 
\[
A := \sum_{i = 1}^n \lambda_i \cL{i}; \quad B := 1 + \sum_{i = 1}^n \lambda_i; \quad \lambda_i = \frac{p_i}{1 - q_i}; \quad z_i = \frac{p_i}{1 - q_i - y_i} - \frac{p_i}{1 - q_i} \geq 0;
\]
\[
\phi(i) := \begin{cases} \ceff{i} &\textnormal{ if } z_i > 0 \\ 0 &\textnormal{ if } z_i = 0\end{cases};
\]
We therefore define 
\[
\textnormal{utility}^{\textnormal{Agent}}(S) := \frac{ A + \sum_{j \in S} z_j\phi(j)}{B + \sum_{j \in S} z_j} 
\]
\end{lemma}
\begin{proof}
See Appendix~\ref{appendixA}. 
\end{proof}
We note that here we assume that $y_i > 0$ at each state --- that is to say, adopting the platform increases a state's recurrence probability.  This assumption is not necessary, and the general case can be handled in a similar way by modifying the greedy algorithm to pay attention to signs (see Appendix~\ref{appendixA}). We also reiterate that the solution to the original average case MDP problem need not be unique. Therefore, the argmax solution to Equation~\ref{eq:agent-opt} has many potential solutions.

The optimization problem formulated in Lemma~\ref{thm:equiv-ip} can be solved in polynomial time.  The intuitive reason is this:  Looking at the fractional objective function, we note that it is the ratio of two linear functions of the combinatorial (integer) variables implicit in $S$, and such functions are known to be quasiconvex.  It is therefore no huge surprise that a greedy algorithm solves it --- however, the details are rather involved.  Incidentally, one could arrive at the same algorithm by tracing the simplex algorithm on the MDP linear program.

\begin{algorithm}
\DontPrintSemicolon 
\KwIn{Parameters of the Agent's problem: transition probabilities and utility coefficients in and out of the platform.}
\KwOut{An optimal subset $S \subseteq [n]$ of states where the Agent accepts the platform.}
Initialize $S := \{\}$\;
\For{$k \in [n]$ sorted\footnote{Note the sort order may not be unique in case of ties.} from largest to smallest $\phi(k)$} {
  \uIf{$\textnormal{utility}^{\textnormal{Agent}}(S) < \phi(k)$} {
    Update $S := S \cup \left\{k\right\}$\;
  }
  \Else{
    \Return $S$\;
  }
}
\Return{$S$}\;
\caption{{\sc Greedy Algorithm}} 
\label{alg:agent-greedy}
\end{algorithm}

\begin{theorem}\label{thm:greedy-is-opt}
Algorithm~\ref{alg:agent-greedy} returns
\[
S^* \in \argmax{S \subseteq [n]} \textnormal{utility}^{\textnormal{Agent}}(S) 
\]
That is, the policy 
\[
\pi(s) = \begin{cases} a^1 &\textnormal{ if } s \in S^* \\ a^0 &\textnormal{ o.w.} \end{cases}
\]
is an optimal policy. Here, $a^1$ and $a^0$ refer to the actions available to the Agent's MDP: "accept platform" is $a^1$ and "do not accept platform" is $a^0$. 
\end{theorem}

Before we prove the theorem, we give a useful definition and a lemma. 
 
\begin{definition}[Prefix policy]
We say a policy $S$ is prefix if the states in the policy are the first $m$ states
in order sorted by $\phi$, for some value of $m$. 
\end{definition}

\begin{lemma}[Mediant Inequality]\label{lemma:1}
\[
\frac{x}{y} < \frac{r}{s} \iff \frac{x}{y} < \frac{x + r}{y + s} < \frac{r}{s}\quad \textnormal{ where } y, s > 0.
\]
\end{lemma}
\begin{proof}
Since $y, s > 0$ and thus $y + s > 0$, cross-multiply and simplify to get the desired inequalities.\qed
\end{proof}

With this lemma in hand, we prove Theorem~\ref{thm:greedy-is-opt}.
\begin{proof}[Proof of Theorem~\ref{thm:greedy-is-opt}]
We can prove optimality in two steps. 
\begin{enumerate}
    \item First we will show that any non-prefix policy is dominated by a prefix policy. Thus an optimal policy must be prefix. 
    \item Then, we show that the greedy algorithm necessarily finds a best prefix policy (e.g., an optimal stopping point). 
\end{enumerate}

We begin with the first step. Suppose we have a non-prefix policy $S$. Let state $\ell \in [n]$ be a ``missing piece'' (e.g., if we index by sorted order and $S$ contained $1, 2, 4, 5, 7$, missing pieces would be $3$ and $6$). This $\ell \in [n]$ necessarily exists since $S$ is non-prefix. 
Now there are two cases. 

\paragraph{Case $1$:}
\[
\textnormal{utility}^{\textnormal{Agent}}(S) < \phi(\ell)
\]
We apply Lemma~\ref{lemma:1} 
to show that adding state $\ell$ results in improvement in the objective. 

\paragraph{Case $2$:}
\[
\textnormal{utility}^{\textnormal{Agent}}(S) \geq \phi(\ell)
\]
Here we show that removing all states $k \in S$ where $\phi(k) < \phi(\ell)$ improves the objective. If equality holds, then it does not matter whether we add the state to the objective, so for simplicity, we terminate at equality. From the assumption and the definition of $\phi$, we have 
\[
\textnormal{utility}^{\textnormal{Agent}}(S) \geq \phi(\ell) > \phi(k)
\]
for all such $k$. By Lemma~\ref{lemma:1}, removing state $k$ increases the ratio, e.g.
\[
\textnormal{utility}^{\textnormal{Agent}}(S\setminus\{k\}) > \textnormal{utility}^{\textnormal{Agent}}(S) \geq \phi(\ell) > \phi(k)
\]
The same argument applies to all $k' \in S\setminus\{k\}$ such that $\phi(k') \leq \phi(k)$ as well. Therefore, we can remove all the $k'$ with score less than the score of $\ell$ and improve the objective. 

After a single round of considering a missing state $\ell$ (where either case $1$ or case $2$ applies), we produce a new $S'$, which can again be non-prefix. However, the maximum index present in the new $S'$ has either decreased (if the second case happened and we deleted everything worse than $\ell$) or we filled in the missing state $\ell$. Either way, the number of missing pieces has strictly decreased and we have added no new missing states. Using induction on the number of missing pieces proves that iterating over all original missing pieces will ``fill in all the gaps'' and produce a prefix policy $S^*$ which is strictly better than the original non-prefix policy $S$. 

Finally, we show the greedy algorithm selects an optimal prefix policy. Let the output of the greedy algorithm be $\hat{S}$. The desired result directly follows since if the next state $\ell$ satisfies 
\[
\textnormal{utility}^{\textnormal{Agent}}(\hat{S}) \geq \phi(\ell)
\]
and is not selected, since all smaller states (sorted by $\phi$) are less than or equal to $\phi(\ell)$, any prefix subset of the smaller states is an effective average which is $\leq \phi(\ell)$, and any prefix subset of future states is worse off. Thus, the greedy algorithm produces a maximal solution. 
\end{proof}



%% file: designer.tex
\section{Flower Case: The Designer's Problem} \label{sec:designer}
We now consider the Designer's problem (Definition~\ref{def:PDP}) in the special case where the Agent lives in the flower MDP (Definition~\ref{def:agent-mdp}). Under this assumption on the Agent's MDP, it will be possible to give an FPTAS for the Designer's problem, due to the additional structure imposed in this setting. 
Let $\textnormal{Agent}(S)$ denote the subset of states that the Agent adopts when the Designer offers platforms for the subset $S$ of states.
Given the results of Section~\ref{sec:agent}, the fraction of the time that the Agent spends in state $i \in \textnormal{Agent}(S)$ is 
\[
\frac{\frac{p_i}{1-q_i-y_i}}{B+\sum_{i \in \textnormal{Agent}(S)} z_i}
\]
using the notation of Section 2 given in Definition~\ref{def:agent-mdp} and Lemma~\ref{thm:equiv-ip}
(see Appendix A for the  stationary distribution of the Markov chain).
Thus, we can simplify the expression for the Designer's profit function:
\[
\textnormal{profit}(S) := \frac{ \sum_{i \in \textnormal{Agent}(S)} d_i\cdot\frac{p_i}{1 - q_i - y_i}}{B + \sum_{i \in \textnormal{Agent}(S)} z_i} - \sum_{i \in S} \textnormal{cost}_i
\]
Call a set $S$ of states {\em feasible} if $\textnormal{Agent}(S)= S$. Since the Agent's response is completely anticipated by the Designer (Agent's parameters are known to the Designer, and the Designer can therefore simulate the greedy algorithm from Section~\ref{sec:agent}), only feasible sets $S$ need be considered.

A few additional properties result after we specialize to the flower MDP setting. It is easy to see from the definition and the greedy algorithm of Section~\ref{sec:agent} that if a set $S $ is feasible then so are all its subsets. If for some state $i$, $\textnormal{profit}(\{i\}) \leq 0$, then
it follows that for all sets $S$ that contain $i$ we have
$\textnormal{profit}(S) \leq \textnormal{profit}(S-\{j\})$.
Hence, there is no reason to build a platform at $i$,
and we can ignore $i$.
Thus, we may restrict our attention to the states $i$ such that
$\textnormal{profit}(\{i\}) > 0$.
We may assume also without loss of generality that every state $i$ by itself is feasible: If $\{i\}$ is not feasible, then neither is any set that contains $i$, therefore we can ignore $i$.

We now add a few more assumptions to ensure tractability. Let $K = \max_i \textnormal{profit}(\{i\})$.
It is easy to see that for any set $S$, $\textnormal{profit}(S) \leq \sum_{i \in S} \textnormal{profit}(\{i\})$.
Therefore, the optimal profit $OPT$ is at most $nK$ and at least $K$.
We will also assume that the cost $\textnormal{cost}_i$ of 
building a platform at any site $i$
is not astronomically larger than the anticipated optimal profit,
specifically we assume $\textnormal{cost}_i \leq rK$
for some polynomially bounded factor $r$. Furthermore, and importantly, for our dynamic programming FPTAS to work in polynomial time, a discretization assumption is necessary. For each state $i$, the platform available will change (increase or decrease) the term $\frac{p_i}{1-q_i}$ appearing in the numerator and the denominator of the Agent's objective by an additive $z_i$.  {\em We assume that all these $z_i$'s are multiples of the same small constant $\delta$} (think of $\delta$ as $1\%$).  This assumption means that there are $O\left(\frac{1}{\delta}\right)$ possible values of the denominator, and ensures the dynamic programming is polynomial-time. One should think of this maneuver as one of the compromises (in addition to accepting a slightly suboptimal solution) for the approximation of the whole problem.  We suspect that the problem has no FPTAS without this assumption, although there is a pseudo-polytime algorithm. 

\subsection{FPTAS for the PDP}
The Platform Design Problem is approximately solvable in polynomial time; in this section we present a FPTAS which returns a $(1 - \eps)$-approximate solution. Our approach is inspired by the FPTAS for knapsack presented in \citet{Ibarra75}. We also note that our algorithm relies on the structure of the greedy algorithm presented in Algorithm~\ref{alg:agent-greedy}. 

The algorithm uses dynamic programming.
It employs a 3-dimensional hash table, called SET,
into which the sets under consideration are being hashed.
The hash function has three components that correspond to the
following components of the profit of the set,  scaled and rounded appropriately to integers:
(1) the whole profit $\textnormal{profit}(S)$,
(2) the first term in the profit, denoted $P_1(S)$, and
(3) the denominator of the first term, denoted ${\bf D}(S)$ (which note, is also
the denominator of the Agent's objective function).
We use ${\bf N}(S)$ to denote the numerator of the
Agent's objective function.

\begin{algorithm}
\DontPrintSemicolon 
\KwIn{The parameters of the PDP: transition probabilities, utility and cost coefficients for the Agent and the Designer, and small positive reals $\eps, \delta$}
\KwOut{A $(1-\eps)$-approximately optimal subset of states $S^*$ for which to deploy platforms.}
\textbf{N}(S) and \textbf{D}(S) denote the numerator and the denominator of the Agent's objective function, with the constant terms omitted\;
${P_1}(S)$ denotes the first term in the Designer's profit function\;
\textnormal{SET} is a hash table of subsets of $[n]$ indexed by triples of integers\;
The hash function is $\textnormal{hash}(S) := \left(\lceil \frac{\textnormal{profit}(S)}{\epsilon K/2n}\rceil,\lceil \frac{P_1(S)}{\epsilon K/2n}\rceil, \textbf{D}(S)/\delta \right)$\;
Initialize the hash table \textnormal{SET} to contain only the empty set in the bin $(0,0,0)$\;
\For{$k \in [n]$} {
    \For{$S \in$ \textnormal{SET} \textnormal{ in lexicographic order}}{
        $S' := S \cup \{k\}$\;
        \If{Agent will adopt all platforms in $S'$ and $\textnormal{profit}(S') > 0$}{
            \uIf{\textnormal{hash}$(S') \in \textnormal{SET}$}{
                $\hat{S} := \textnormal{SET}[\textnormal{hash}(S')]$\;
                \If{$\textbf{N}(\hat{S}) > \textbf{N}(S')$}{
                      $\textnormal{SET}[\textnormal{hash}(S')] := S'$\;
                }
            }
            \uElse{
                $\textnormal{SET}[\textnormal{hash}(S')] := S'$\;
            }
        }
    }
}
\Return{the set $S$ in the hash table with largest first hash value}\;
\caption{{\sc Designer's FPTAS for the PDP}} 
\label{alg:designer-FPTAS}
\end{algorithm}

\begin{lemma} \label{lemma:invariance} 
Let $S,S' \subseteq [k]$ be two sets that hash in the same bin and
suppose that $\textbf{N}(S) \leq \textbf{N}(S')$.
Then for every set $T \subseteq \{k+1, \ldots, n \}$,
if $S' \cup T$ is feasible then $S \cup T$ is also feasible,
and $\textnormal{profit}(S \cup T) \geq \textnormal{profit}(S' \cup T)- \epsilon K/n$.
\end{lemma}
\begin{proof}
Proved in Appendix~\ref{appendixB}.
\end{proof}


\begin{lemma} \label{lemma:induction}
For every $k =0, 1, \ldots, n$, after the $k^{th}$ iteration of the loop, there is a set $S$ in the hash table that
can be extended with elements from $\{ k+1, \ldots, n \}$ to a feasible set that has profit $\geq \textnormal{OPT}- \eps k \cdot K/n$.
\end{lemma}
\begin{proof}
Proved in Appendix~\ref{appendixB}.
\end{proof}

\begin{theorem} \label{thm:FPTAS-pdp}
Algorithm~\ref{alg:designer-FPTAS} is a FPTAS for the Platform Design Problem. 
\end{theorem}
\begin{proof}
Lemma~\ref{lemma:induction} for $k=n$ tells us that at the end, the table contains a set $S$ whose profit is within $\epsilon K$ of OPT.
Since $OPT \geq K$, the profit of $S$ is at least $(1 -\epsilon) OPT$.

Regarding the complexity of the algorithm,
note that the three dimensions of the hash table have respectively size $O(n^2/\epsilon)$
(since the maximum profit is at most $nK$),
$O(r n^2/\epsilon)$, and $n/\delta$.
In every iteration the algorithm spends time proportional to the number of sets stored in the table. In particular, the algorithm only needs linear time to check the feasibility of each $S'$ as well as calculate $\textbf{N}(S')$ and profit$(S')$. Thus, the total time is polynomial in $n$ and $\frac{1}{\epsilon}$.
\end{proof}

It turns out that an FPTAS is the best we could hope for, even if all $z_i=1$:



\begin{theorem}\label{thm:hardness}
The PDP in the flower case is NP-complete. 
\end{theorem}
\begin{proof}
Proved in Appendix~\ref{appendixB}.
\end{proof}

%% file: manyagents.tex
\section{An Algorithm for Many Agents} \label{sec.manyAgents}

We generalize the FPTAS of the previous section to the case with $k$ Agents, each with their own flower Markov chain. Our algorithm is polynomial runtime for constant $k$, and is exponential if $k$ is allowed to vary --- hence, it is perhaps more natural to think of $k$ as the number of types in a finite-support distribution of agents, where the number of types may naturally be a small constant in settings of interest. 

We will use the notation of Section \ref{sec:agent}, with an additional
subscript $i$ for each Agent $i$.
Thus, for example $p_{ij}, q_{ij}, y_{ij};  j \in [n]$  denote the parameters of the Markov chain of Agent $i$,  $\phi_i(j)$ denotes
the potential of state $j$ for Agent $i$.
The utility of Agent $i$ if he adopts the platforms in a set $S$
of states is  
$u_i(S) = \frac{A_i + \sum_{j \in S} z_{ij} \phi_i(j)}{B_i + \sum_{j \in S} z_{ij}}$.

The Designer will offer platforms for a set $S$ of states.
If the platform at state $j$ is adopted by Agent $i$, then
the Designer gets reward at a rate $d_{ij}$, i.e. gets
reward equal to $d_{ij}$ times the fraction of the time
that Agent $i$ spends at state $j$.
The cost of building a platform for state $j$ is $\textnormal{cost}_j$.
Thus, the Designer's profit function is:
\[
\textnormal{profit}(S) := \sum_i \frac{ \sum_{j \in \textnormal{Agent}_i(S)} d_{ij}\cdot\frac{p_{ij}}{1 - q_{ij} - y_{ij}}}{B_i + \sum_{l \in \textnormal{Agent}_i(S)} z_{il}} - \sum_{j \in S} \textnormal{cost}_j
\]
where $\textnormal{Agent}_i(S)$ is the set of states
that Agent $i$ chooses when offered platforms in the set $S$,
i.e. the set chosen by the greedy algorithm of Section 2.

We will assume in this section that, besides the parameters $z_{ij}$, also
the potentials $\phi_i(j)$ are quantized.
That is, we assume that both the $z_{ij}$'s and the
$\phi_i(j)$'s are  polynomially bounded integer multiples of some small amounts; i.e., 
each $z_{ij}$ is of the form $l_{ij} \delta$ for some integer
$l_{ij} \leq M$ and some $\delta$, with $M$ polynomially bounded, and similarly
each  $\phi_i(j)$ is of the form $l'_{ij} \delta'$ for some 
integer $l'_{ij} \leq M$ and some $\delta'$. 
This implies then that the numerator and denominator 
of the utility function $u_i(S)$ of an Agent for 
the various sets $S$ can take a polynomial number of possible values.
We will show that under this assumption, the optimal
solution can be computed in polynomial time for a fixed number of Agents.

For each Agent $i$, let $\Phi_i = \{ \phi_i(j) | j \in [n] \} \cup \{\infty\}$.
Let ${\cal D}_i = \{ B_i + l \delta | l \in [nM] \}$, 
${\cal N}_i = \{ A_i + l \delta \delta' | l \in [nM^2] \}$.
Note that $|\Phi_i| , |{\cal D}_i|, |{\cal N}_i|$ are polynomially bounded
by our assumption. By the definitions,
for every subset $S$ of states, the numerator of
the utility $u_i(S)$ is in ${\cal N}_i$ and the denominator is in ${\cal D}_i$. 
Let $\Phi = \Pi_{i=1}^k \Phi_i$,
${\cal D} = \Pi_{i=1}^k {\cal D}_i$, 
and ${\cal N} = \Pi_{i=1}^k {\cal N}_i$.

For any $\theta_i \in \Phi_i$, 
define $Q_i(\theta_i) = \{j | \phi_i(j) \geq \theta_i \}$.
For any tuple $\theta \in \Phi$ and tuple $D \in {\cal D}$,
define a corresponding value coefficient $c_j$ for state $j$
to be 
$$c_j(\theta,D) = \sum_{i: j \in Q_i(\theta_i)}  \frac{d_{ij} p_{ij}}{(1-q_{ij}-y_{ij})D_i }  -\textnormal{cost}_j$$

That is, the summation in the above formula includes only
those $i \in [k]$ such that $j \in Q_i(\theta_i)$.

The algorithm is given below.
It uses dynamic programming.
For every tuple $\theta \in \Phi$ and  $D \in {\cal D}$,
it computes an optimal set $S$ such that, if the Designer offers
platforms in the subset $S$ of states, then Agent $i$
selects all states of $S$ that have $\phi_i(j) \geq \theta_i$
(i.e. $\textnormal{Agent}_i(S) = S \cap Q_i(\theta)$),
and the denominator of $u_i(S)$ is $D_i$.
The algorithm then returns the best $S$ that it finds
over all tuples $\theta \in \Phi$ and  $D \in {\cal D}$.

For every tuple $\theta \in \Phi$ and  $D \in {\cal D}$,
the algorithm processes the states in (arbitrary) order $1, \ldots, n$.
It employs a hash table $H$ indexed by two $k$-tuples $a, b$
of integers, where $a \in ([M^2])^k$, $b \in M^k$, represent
respectively the integer parts of the numerators and denominators
of the utility functions of the Agents for a set.
Each entry $H(a,b)$ of the hash table is either empty or
contains a subset of the states processed so far that hashes into this slot.
A set $S$ of states hashes into the slot $(a,b)$ where
$a_i = \sum_{j \in S \cap Q_i(\theta_i)} \frac{z_{ij} \phi_i(j)}{\delta \delta'}$ and 
$b_i = \sum_{j \in S \cap Q_i(\theta_i)} \frac{z_{ij}}{\delta}$
for all Agents $i \in [k]$.
We can implement $H$ by a search data structure that contains
only the slots $(a,b)$ that are nonempty and for each one of them
records the corresponding set $S$ and its value with respect
to $\theta, D$.
We define the {\em value} of a set $S$ to be
$\textnormal{value}_{(\theta,D)}(S) = \sum_{j \in S} c_j(\theta,D)$.

Initially the hash table $H$ contains only the empty set in the slot $(0,0)$.
After the algorithm processes all the states, it
examines the slots $(a,b)$ that are consistent with the
pair $(\theta, D)$ in the following sense.
For each $i \in [k]$, let 
$\theta'_i =\max_{j} \{ \phi_i(j) | \phi_i(j) < \theta_i \}$,
i.e. the next smaller value of a potential $\phi_i(j)$
after $\theta_i$; if $\theta_i = \infty$, then $\theta'_i$ is
the maximum $\phi_i(j)$, and if $\theta_i$ is the smallest
$\phi_i(j)$, then set $\theta'_i =-1$.
We say that the pair $(a,b)$ is {\em consistent} with $(\theta,D)$
if $D_i = B_i + b_i \delta$, and 
$\theta_i > \frac{A_i + a_i \delta \delta'}{D_i} \geq \theta'_i$ for all $i \in [k]$.
The algorithm sets $S(\theta,D)$ to be the set $H[a,b]$
with the largest $\textnormal{value}_{(\theta,D)}$
among the consistent slots $(a,b)$.
At the end, 
after the algorithm has processed all the pairs $(\theta,D)$,
it returns the set $S(\theta,D)$ with the largest value.

\begin{algorithm}
\DontPrintSemicolon 
\KwIn{The parameters of the PDP: transition probabilities, utility and cost coefficients for the Agents and the Designer}
\KwOut{An optimal subset of states $S^*$ for which to deploy platforms.}

\For{\textnormal{each } $\theta \in \Phi, D \in {\cal D}$} {
Initialize the hash table $H$ to contain only the empty set in the slot $(0,0)$\;
	\For{$t \in [n]$} {
    	\For{ \textnormal{each nonempty slot } $(a,b)$ \textnormal{of } $H$}{
       	 	$S= H(a,b)$;
        	$S' := S \cup \{t\}$\;
            \uIf{\textnormal{hash}$(S') \in \textnormal{H}$}{
                $\hat{S} := H[\textnormal{hash}(S')]$\;
                \If{$\textnormal{value}_{(\theta,D)}(\hat{S}) < \textnormal{value}_{(\theta,D)}(S')$}{
                      $H[\textnormal{hash}(S')] := S'$\;
                }
            }
            \uElse{
                $H[\textnormal{hash}(S')] := S'$\;
            }
        }
    }

$S(\theta,D) := \argmax{H(a,b)} \{ \textnormal{value}_{(\theta,D)}(H(a,b))~~ |~~(a,b) \textnormal{ is consistent with }(\theta,D) \} $\;
}
\Return{ $\argmax{S(\theta,D)} \{ \textnormal{value}_{(\theta,D)}(S(\theta,D)) ~~|~~ (\theta, D) \in (\Phi, {\cal D}) \} $ }\;
\caption{{\sc Designer's multiAgent algorithm for the PDP}} 
\label{alg:Designer-multi}
\end{algorithm}

\begin{lemma} \label{lem.multi1}
For every pair $(\theta,D) \in (\Phi, {\cal D})$,
if a set $S$ hashes into a slot $(a,b)$ that is consistent
with $(\theta,D)$, 
then $\textnormal{value}_{(\theta,D)}(S) = \textnormal{profit}(S)$.
In particular, the set $S(\theta,D)$ selected by the algorithm (if any) 
satisfies $\textnormal{value}_{(\theta,D)}(S(\theta,D)) = \textnormal{profit} ( S(\theta,D))$.
\end{lemma}
\begin{proof}
Proved in Appendix \ref{appendixC}.
\end{proof}


We remark incidentally that if a slot is not consistent 
with $(\theta,D)$,
then the value of its set may not be equal to its profit
(the profit may be higher or lower).

\begin{lemma} \label{lem.multi2}
Let $S^*$ be an optimal solution to the Platform Design Problem,
and let $\theta_i = \min_{j \in \textnormal{Agent}_i(S^*)} \{ \phi_i(j)\}$, $D_i = B_i + \sum_{j \in \textnormal{Agent}_i(S^*)} z_{ij}$.
Then, in the iteration for the pair $(\theta,D)$, the algorithm
selects a set $S(\theta,D)$, and the set has
$\textnormal{profit}(S(\theta,D)) \geq \textnormal{profit}(S^*)$.
\end{lemma}
\begin{proof}
Proved in Appendix \ref{appendixC}.
\end{proof}

Optimality follows directly from the lemmas.

\begin{theorem}
\label{thm:multi-Agent-PDP} 
Algorithm \ref{alg:Designer-multi} computes an optimal solution to
the Designer's problem. It runs in polynomial time for fixed number of
Agents, under the stated assumptions on the input parameters.
\end{theorem}

For one Agent we gave in the previous Section an FPTAS under the weaker
assumption that the $z$ parameters are polynomially bounded, but 
not necessarily the potentials $\phi$. We can show that there is no such FPTAS for two Agents, if the $\phi$ are not restricted.

\begin{theorem} \label{thm:hard2}
Unless P=NP, there is no FPTAS for the Designer's problem with
two Agents if the $\phi_i(j)$ are not restricted to be polynomially bounded.
\end{theorem}
\begin{proof}
Proved in Appendix~\ref{appendixC}.
\end{proof}

%% file: competitive.tex
\section{The PDP Problem in a Competitive Setting} \label{sec:competitive}
 Many platform designers compete in the world today, and it is of great interest to understand the interaction of two or more platform designers with agents.  In this section (in summary) and continuing in the Appendix~\ref{appendixD} (in full detail), we confront the algorithmic problems involved with designer competition, such as the {\em best response problem:} if a Designer is confronted with a situation in which other designers have already deployed several platforms at various states, which platforms should this Designer deploy?  But even before this, we need to address the following:
 
\subsection{The Agent's Problem with Multiple Platforms per State}
We are given a set of available platforms, where each platform
is associated with one state of the flower MDP.
For each available platform $j$, we are given the associated agent's reward and the change in the transition probabilities of the state;
these induce the corresponding parameters $z_j$ and $\phi(j)$ as
in Section \ref{sec:agent}.
The agent will select a subset $S$ of platforms that contains at most
one platform for each state; call such a set `feasible'.
The agent's utility $u(S)$ for a feasible set $S$ is  
$u(S)= \frac{A+\sum_{j \in S}z_j \phi(j) }{B + \sum_{j \in S} z_j}$.
The agent's objective is to select a feasible set $S$ that maximizes $u(S)$.

We first show that a feasible solution $S$ is optimal if and only if it cannot be improved by (1) removing a platform $j$ from $S$, or (2) adding a platform to $S$ or (3) swapping one platform $j \in S$ for another platform $j' \notin S$ associated with the same state.
The first type of change is beneficial if $\phi(j) < u(S)$,
the second type if $\phi(j) > u(S)$ and $S$ does not contain another platform for the same state. The third type of change is beneficial if 
$z_j = z_{j'}$ and $\phi(j) < \phi(j')$, or
$z_j \neq z_{j'}$ and $z_{j'} -z_j$ has the same sign as
the quantity 
$\rho(j,j')-u(S)$, where  $\rho(j,j')= \frac{z_{j'}\phi(j')-z_j \phi(j)}{z_{j'} - z_j}$.

We then identify and remove platforms that are dominated
and thus redundant.
For every state $s$, the nonredundant platforms form a sequence
$j_1, j_2, \ldots, j_k$ which is decreasing in potential,
and increasing in the value of $z$ and
of $z \phi$.
If we map every platform $j$ of the state $s$ to a point 
$(z_{j}, z_j \phi(j))$ on the plane, the sequence yields
a piecewise-linear concave curve $P_s$.
Note that the ratio $\rho(j,j')$ for two platforms $j, j'$ is the
slope of segment $(p_j,p_{j'})$. The slopes are decreasing along
the curve.
For every nonredundant platform $j$ we use $prev(j)$ to denote the
previous platform in the sequence for its state (if it exists, i.e.
$prev(j_i) = j_{i-1}$ if $i>1$), and $next(j)$ the next platform (if $i<k$).
We show the following optimality criterion.

\begin{lemma}\label{lem:multi-criterion}
Let $S$ be a feasible set of nonredundant platforms.
The set $S$ is optimal if and only if for every state $s$, either
(1) $S$ does not contain any platform for $s$ and all platforms $j$ for $s$ have potential $\phi(j) \leq u(S)$, or
(2) the platform $j \in S$ for state $s$ satisfies 
(i) $\rho(prev(j),j) \geq u(S)$ if $prev(j)$ exists, else $\phi(j) \geq u(S)$, and (ii) $\rho(next(j),j) \leq u(S)$ if $next(j)$ exists.
\end{lemma}

We compute an optimal solution using a greedy algorithm
with a different parameter $\psi(j)$ for each nonredundant platform.
If $j$ is the first nonredundant platform in the sequence 
for its state, then
set $\psi(j) = \phi(j)$, otherwise set $\psi(j)= \rho(prev(j),j)$.
Note that $\psi(j) \leq \phi(j)$ for all $j$.
The algorithm is given below.

\begin{algorithm}
\DontPrintSemicolon 
\KwIn{Parameters of the Agent's problem: transition probabilities and utility coefficients in and out for all platforms.}
\KwOut{An optimal feasible subset $S$ of platforms.}
Remove redundant platforms for each state\;
Compute the parameters $\psi$ for the (nonredundant) platforms\;
Sort the platforms in decreasing order $\psi(j)$\;
Initialize $S := \{\}$\;
\For{each platform $j$ in decreasing order of $\psi(j)$} {
  \uIf{$\psi(j) \leq u(S)$} {
  	\Return $S$\;
  }
   \Else{
   	\uIf{$j$ is the first platform for its state} {
    	Update $S := S \cup \left\{j\right\}$\;
  	}
	\Else{
		
			Update $S := S \cup \left\{j\right\} \setminus \{prev(j)\}$\;
    }
  }
}
\Return{$S$}\;
\caption{{\sc Multi-platform Agent's Algorithm}} 
\label{alg:agent-multi}
\end{algorithm}

\begin{theorem}\label{thm:agent-multi-opt}
Algorithm~\ref{alg:agent-multi} returns an optimal feasible solution.
The algorithm runs in time $O(n+m \log m)$, where
$n$ is the number of states and $m$ is the number of platforms.
\end{theorem}

\subsection{The Designer Problem in a Competitive Setting}

Consider a
Designer choosing which platforms to offer when there
are already in the market available platforms from other
providers. We extend the algorithm of Section \ref{sec.manyAgents}
to this setting.
We have $k$ Agents, each with their own flower Markov chain
on the same state set (but different transition probabilities).
There is a set of existing available platforms (offered by other providers).
The Designer can build a platform for each state, and wants to select an optimal subset of platforms that maximizes the profit.

We use the notation of Section \ref{sec.manyAgents} and
the previous subsection.
We assume that the parameters $z_{ij}$ and $\phi_i(j)$ for agent $i$
and platform $j$ are quantized.
That is, we assume that  
each $z_{ij}=l_{ij} \delta$ for some integer
$l_{ij} \leq M$ and some $\delta$, with $M$ polynomially bounded, and similarly
each  $\phi_i(j)=l'_{ij} \delta'$ for some 
integer $l'_{ij} \leq M$ and some $\delta'$. 
We use essentially the algorithm of Section \ref{sec.manyAgents} 
but modify the definitions of several key concepts and quantities
for the more general setting.
In particular, we modify the definitions of $\Phi_i$,
of $Q_i(\theta_i)$, the hashing function that hashes 
subsets of platforms to pairs $(a,b)$ of tuples, 
and the notion of consistency of a slot $(a,b)$ with 
a pair $(\theta, D) \in (\Phi,{\cal D})$ (see the Appendix).
We prove analogues of the lemmas of Section \ref{sec.manyAgents},
which yield the following:

\begin{theorem}
\label{thm:multi-Designer-PDP} 
The extension of Algorithm \ref{alg:Designer-multi} 
to the competitive setting computes an optimal solution to
the Designer's problem. It runs in polynomial time for fixed number of
Agents, under the stated assumptions on the input parameters.
\end{theorem}

The algorithm can be extended to the case where the Designer can choose between multiple possible platforms for each state.

%% file: discussion.tex
\section{Discussion and Future Work}
\label{sec:discussion}
We believe that we have barely scratched the surface of a very important subject: the economic/mathematical/algorithmic modeling of the interactions between Designers of on-line platforms and the consumers of on-line services/producers of data.  Our model captures a few of the important aspects of this complex environment: the way adoption of services affects both the user's activities and the user's enjoyment of these activities, while it enhances the Designer's revenue in ways that depend on the time spent and activities performed on the platform; the nature of the Designer's profit (revenue from the acquisition of data pertaining to the user minus the significant development costs); the fact that multiple platforms, even by the same Designer, compete for the user's attention and use; the nature of some of the user's dilemmas (chief among them: surrender privacy for increased efficiency and/or enjoyment?).  A simplified model of these aspects (the flower chain, linearity of utilities) is a tractable bi-level optimization problem. 
However, there are many effects that our current model does not capture, which are quite interesting for future research: for instance, as a sample, we may want to model time dependencies in profits, rewards, and costs, scaling effects for the platform designers due to increasing numbers of users, synergistic effects for the agents who may adopt suites of platforms (for instance, adopting all of the Google suite of products may provide more benefit than using different providers for each service), and potential network effects involved in influencing agent behavior when there are many agents. 

We believe that intractability (both analytical and  computational) lurks in many of the possible immediate generalizations of this model --- for example, to undiscretized coefficients, to Markov chains more general than the flower, or to more complex objectives than linear (such as the addition of an entropy regularizer to the objectives of both the Agent and the Designer --- an especially tempting variant to consider in this particular problem).  We believe that more ambitious problem formulations in these directions may need to further simplify the other aspects of the model in this paper to become tractable. 

On the other hand, we also believe that any form of intractability of the Designer's problem is arguably {\em affordable.}  Our dynamic programming FPTAS would likely not generalize to more general contexts --- such as those involving complex chains, nonlinear objectives, many Designers, learning of the statistics of the Agents' parameters etc, see below --- but the alternative exhaustive algorithm, with its rather benign exponential dependency on $n$, the number of platforms, is extremely realistic in this context.  We believe that the true challenges in generalizing our results are challenges of formulation and modeling.

Superficially, platform design resembles Mechanism Design (MD) \citep{Myerson83}, but the essence of much of MD is that the Designer knows only {\em statistics} of the Agent's characteristics and designs the mechanism to optimize revenue over all possible eventualities by incentivizing the Agent to implicitly reveal their type; and this essence is missing in the PDP.  In the on-line platform environment, the subject of incentives for type revelation and truthfulness is rather clearly related to the {\em personalization} of the platform, and we believe that a generalization of our model will have to address this important issue and aspect of platform design.  

In the present first brush at platform design, we have abstracted the PDP in terms of a single Agent --- a maneuver and methodology familiar from Economics ---, and next ventured to the case of a few Agent types.  But of course the motivating environment involves myriads of atypical Agents.  
The Agent's statistics must be {\em learned with high accuracy}, and we believe that this can lead to the development of novel aspects of Machine Learning. The learning nature of platform design resembles interactive learning (the learner and the Agent whose parameters are being learned interact closely, and the learner can easily experiment with variants of the platform), and also has certain characteristics of learning from revealed preferences, see e.g.~\citet{Zadimoghaddam12}. We believe that a wealth of novel and intriguing technical problems within Learning Theory and Machine Learning lie in this direction, and can build on recent work in the intersection of these areas with Algorithmic Mechanism Design and Learning in Games \citep{Haghtalab18, Foster16}. Of course, recent cautionary results on the limitations of optimization by samples \citep{Balkanski17a, Balkanski17b} come to mind as well.

Regarding the important subject of strategic interactions between designers, we have not addressed the equilibrium problem --- beyond the best-response algorithm.  We can show (see Appendix~\ref{appendixF}) that a pure Nash equilibrium may not exist even in the flower setting, and we conjecture that finding a pure equilibrium is $\Sigma_2$-complete.
But perhaps the most interesting strategic questions go beyond the model of this paper:  How are Designers incentivized by the competition to design and deploy platforms that are more beneficial to the Agents than in the monopolistic situation?

Finally --- and almost needless to say --- the subject of platform design, as circumscribed in this paper, is crying out for treatment from the point of view of the exploding literature on {\em  ethics, fairness, and privacy} in algorithm design --- see for example \citet{Dwork12, Kleinberg01, Gemici18} among many other important works --- and exposes new aspects of today's algorithmic environment to these important considerations and emerging methodologies. The PDP defines an environment where privacy and fairness concerns are ubiquitous and paramount.  
Understanding what kinds of social, economic, regulatory, and technological interventions may result in fairer outcomes of platform design is an important direction of future work.

%% file: acknowledgements.tex
\section{Acknowledgements} 
\label{sec:ack} 

We thank John Tsitsklis, Eva Tardos, and Yang Cai for helpful conversations during the development of this work. K. Vodrahalli acknowledges support from an NSF Graduate Fellowship. 

%% file: appendixA.tex
\section{Proofs for Section~\ref{sec:agent}} \label{appendixA}
\subsection{Proof of the Agent's Objective in the Flower MDP}
\noindent
{\bf Lemma \ref{thm:equiv-ip}}.
The Agent's objective for an optimal policy defined in Section~\ref{sec.general_pdp} can be re-written as the following optimization in the special case of the flower MDP (Definition~\ref{def:agent-mdp}):
\begin{align}
\begin{split}
\argmax{S \subseteq [n]} \frac{ A + \sum_{j \in S} z_j\phi(j)}{B + \sum_{j \in S} z_j} 
\end{split}
\end{align}
where 
\[
A := \sum_{i = 1}^n \lambda_i \cL{i}; \quad B := 1 + \sum_{i = 1}^n \lambda_i; \quad \lambda_i = \frac{p_i}{1 - q_i}; \quad z_i = \frac{p_i}{1 - q_i - y_i} - \frac{p_i}{1 - q_i} \geq 0;
\]
\[
\phi(i) := \begin{cases} \ceff{i} &\textnormal{ if } z_i > 0 \\ 0 &\textnormal{ if } z_i = 0\end{cases};
\]
We therefore define 
\[
\textnormal{utility}^{\textnormal{Agent}}(S) := \frac{ A + \sum_{j \in S} z_j\phi(j)}{B + \sum_{j \in S} z_j} 
\]

\begin{proof}
First, define the Markov chain transition matrix as the composition of an optimal policy $\rho^*$ and $T$: 
\[
M(i, j) = \sum_{a \in \{a^0, a^1\}} \rho^*(i, a) T(i, a, j)
\]
By the ergodic theorem for irreducible finite-state Markov chains, it is well known that we can express the total reward of an irreducible average-reward MDP as
\[
R(\rho^*) = \langle \pi(\rho^*), r(\rho^*)\rangle 
\]
where $\pi(\rho^*)$ is the stationary distribution that results from playing policy $\rho^*$ and reward vector $r \in \mathbb{R}^{n + 1}$ is fixed for all time, since $\rho^*$ is fixed and the rewards only depend on state and action values (see e.g. \cite{Puterman94, Bertsekas17}).
Let us now calculate $\pi$ and $r$. By solving the balance equations for $M$, elementary algebra shows that the stationary distribution is as follows: 
Define 
\[
x_i(\rho) := \begin{cases} \frac{p_i}{1 - q_i} &\textnormal{ if } \rho(i, a^0) = 1 \\ \frac{p_i}{1 - q_i - y_i} &\textnormal{ if } \rho(i, a^1) = 1 \\ 1 &\textnormal{ if } i = 0 \end{cases}
\]
Then, the stationary distribution is given by 
\[
\pi_i(\rho) := \frac{x_i(\rho)}{\sum_{j = 0}^n x_j(\rho)}
\]
We also note that 
\[
r_i(\rho) := \begin{cases} \cL{i} &\textnormal{ if } \rho(i, a^0) = 1 \\ \cP{i} &\textnormal{ if } \rho(i, a^1) = 1 \\ 0 &\textnormal{ if } i = 0\end{cases}
\]
Translating these values into our objective $R(\rho) = \langle \pi, r\rangle$, we want to solve 
\[
\max_{\rho \in \{0, 1\}^n} \sum_{i = 1}^n \pi_i(\rho)\cdot r_i(\rho) := \max_{\rho \in \{0, 1\}^n} \sum_{i = 1}^n \frac{x_i(\rho)}{\sum_{j = 0}^n x_j(\rho)} \cdot r_i(\rho) = \max_{\rho \in \{0, 1\}^n} \frac{\sum_{i = 1}^n x_i(\rho)\cdot r_i(\rho)}{\sum_{j = 0}^n x_j(\rho)} 
\]
which we note is exactly the objective stated in the theorem.
\end{proof}

\subsection{The Greedy Algorithm when the $y_i$ can be Negative} 

In the case where $y_i$ is allowed to be negative, we need to slightly modify Algorithm~\ref{alg:agent-greedy} as well as its proof. The new algorithm is as shown below.

\begin{algorithm}
\DontPrintSemicolon 
\KwIn{Parameters of the Agent's problem: transition probabilities and utility coefficients in and out of the platform.}
\KwOut{An optimal subset $S \subseteq [n]$ of states where the Agent accepts the platform.}

Initialize $S := \{\}$\;
Divide $[n]$ into two lists $P = \{i: z_i > 0\}; N = \{j: z_j < 0\}$\;
Sort $P$ in order of $\phi(k)$ from largest to smallest\;
Sort $N$ in order of $\phi(k)$ from smallest to largest\;
\For{$k \in P$ } {
    \uIf{$\textnormal{utility}^{\textnormal{Agent}}(S) < \phi(k)$} {
    Update $S := S \cup \left\{k\right\}$\;
  }
  \Else{
    Break\;
  }
}
\For{$k \in N$ } {
    \uIf{$\textnormal{utility}^{\textnormal{Agent}}(S) > \phi(k)$} {
    Update $S := S \cup \left\{k\right\}$\;
  }
  \Else{
    Break\;
  }
}
\Return{$S$}\;
\caption{{\sc General Greedy Algorithm}} 
\label{alg:agent-greedy-all-signs}
\end{algorithm}

\begin{lemma}
\label{lemma:general_ineq_spielman} 
Suppose $s < 0$ and $y + s > 0$. Then
\[
\frac{x}{y} < \frac{r}{s} \iff \frac{r}{s} > \frac{x}{y} > \frac{x + r}{y + s}
\]
and 
\[
\frac{x}{y} > \frac{r}{s} \iff \frac{r}{s} < \frac{x}{y} < \frac{x + r}{y + s}
\]
\end{lemma}
\begin{proof}
For the first inequality, we have that since $s < 0, y > 0$ the LHS implies 
\[
xs > ry \iff xy + xs > xy + ry 
\]
which is $\iff$ the RHS. 
The second inequality follows from the same argument reversed. 
\end{proof} 

\begin{theorem}
Algorithm~\ref{alg:agent-greedy-all-signs} returns an optimal policy when there exist $y_i < 0$. 
\end{theorem}
\begin{proof}
We first sketch the main idea: adding more new states from $P$ is possible when the utility is small, and adding more new states from $N$ is possible when the utility is large, using Lemmas~\ref{lemma:1} and ~\ref{lemma:general_ineq_spielman}. Thus, to maximize utility, first maximize the utility over $P$ (allowing the utility to be as large as possible for adding states from $N$). Adding the additional states from $N$ does not mean there are additional states from $P$ to add, because now the utility is larger than when we stopped adding from $P$, and the prefix arguments from Theorem~\ref{thm:greedy-is-opt} imply we are done.

Now we elaborate on the details.
By Theorem~\ref{thm:greedy-is-opt}, after the first for loop in Algorithm~\ref{alg:agent-greedy-all-signs}, we have an optimal policy over the states in $P$. 
For the states in $N$, first note that we can apply  Lemma~\ref{lemma:general_ineq_spielman} as follows: Choose $x = A + \sum_{j \in S} z_j\phi(j)$, $y = B + \sum_{j \in S} z_j$, $r = \phi(k)z_k$, $s = z_k$.
We will always have $z_k < 0$ for $k \in N$. We will also always have $y + z_k > 0$ for $k \in N$:
We only need to show that $y > |z_k|$. Writing out the expressions from Lemma~\ref{thm:equiv-ip}, we have for any set $T \subseteq [n]$
\[
1 + \sum_{i \in [n]} \lambda_i - \sum_{j \in T} |z_j| = 1 + \sum_{i \in [n]} \frac{p_i}{1 - q_i} - \sum_{j \in T} \frac{p_j}{1 - q_j} + \sum_{j \in T} \frac{p_i}{1 - q_i - y_i} 
\]
\[
= 1 + \sum_{i \not\in T} \frac{p_i}{1 - q_i} + \sum_{j \in T} \frac{p_j}{1 - q_j - y_j} > 0
\]
Thus, the case when $y_k < 0$ (and thus $z_k < 0$) satisfies the conditions of Lemma~\ref{lemma:general_ineq_spielman}.

Now we prove optimality. 
First note that once a subset from $P$ is fixed, the optimal additional subset of states from $N$ can be determined greedily by adding the states with the smallest viable potential $\phi(k)$ first. By Lemma~\ref{lemma:general_ineq_spielman}, if $x/y > \phi(k)$, adding state $k$ increases the utility to $x'/y' > x/y$. Since the utility has increased, we are able to add all states with $\phi(k) < x/y$ and increase the utility -- not adding any of these states to increase the utility when we can results in sub-optimality (note that Algorithm~\ref{alg:agent-greedy-all-signs} indeed ensures all of these states will be added, since we start from the state with the smallest potential). This fact establishes a similar prefix property for adding states from $N$ (see the proof of Theorem~\ref{thm:greedy-is-opt}). Thus, once the greedy algorithm stops, we have an optimal policy with respect to any given fixed subset of states from $P$. 

To conclude the proof, we show that any policy not containing an optimal policy over $P$ is sub-optimal. This fact directly implies that Algorithm~\ref{alg:agent-greedy-all-signs} results in an optimal policy, since it first chooses an optimal subset of $P$, and then adds from $N$ greedily in an optimal manner, as described above. Consider that by Lemma~\ref{lemma:general_ineq_spielman}, if we choose a sub-optimal subset of $P$, $x/y$ will be smaller than if we chose an optimal subset of $P$. Since one can add more states (with larger potential) from $N$ the larger the initial $x/y$, and since to be optimal, one must add \textit{all} states with potential less than $x/y$, the overall utility gained by adding states from $N$ is at most the utility of the optimal subset of $P$ plus the optimally added states from $N$. Thus the policy is sub-optimal unless we choose an optimal subset of $P$ from the start. 
\end{proof}


Note now that Theorems~\ref{thm:FPTAS-pdp}, ~\ref{thm:hardness}, ~\ref{thm:multi-Agent-PDP}, ~\ref{thm:hard2} all go through
identically in the case where there exist states $i$ such that $y_i < 0$. 
Changing the Agent's algorithm from Algorithm~\ref{alg:agent-greedy} to Algorithm~\ref{alg:agent-greedy-all-signs} does not affect the proofs of any of the above theorems. For the FPTAS results, the Designer's algorithm only uses the Agent's algorithm to check feasibility in a black-box fashion, and the sign of $z_k$ in the objective does not affect the analysis of the theorems. 
Finally, the hardness results are unaffected since the constructed hard examples can just choose a case where all $y_i > 0$. 

%% file: appendixB.tex
\section{Proofs for Section~\ref{sec:designer}}
\label{appendixB}


\subsection{Proofs of Lemmas used in Theorem~\ref{thm:FPTAS-pdp}}
We first collect the proofs for the FPTAS. 

\noindent
{\bf Lemma \ref{lemma:invariance}}.
Let $S,S' \subseteq [k]$ be two sets that hash in the same bin and
suppose that $\textbf{N}(S) \leq \textbf{N}(S')$.
Then for every set $T \subseteq \{k+1, \ldots, n \}$,
if $S' \cup T$ is feasible then $S \cup T$ is also feasible,
and $\textnormal{profit}(S \cup T) \geq \textnormal{profit}(S' \cup T)- \epsilon K/n$.
\begin{proof}
Since $S, S'$ hash in the same bin, they have the same denominator
$\textbf{D}(S) = \textbf{D}(S')$. Hence also
$\textbf{D}(S \cup T) = \textbf{D}(S' \cup T)$.
Furthermore, $\textbf{N}(S) \leq \textbf{N}(S')$ implies
that $\textbf{N}(S \cup T) \leq \textbf{N}(S' \cup T)$.
Since $S' \cup T$ is feasible it follows that $S \cup T$ is also feasible.

Consider the difference $\textnormal{profit}(S' \cup T) - \textnormal{profit}(S \cup T)$.
From the formulas it follows that this difference is equal to
$\textnormal{profit}(S' ) - \textnormal{profit}(S)+
(P_1(S)-P_1(S')) \frac{\sum_{i \in T} z_i}{\textbf{D}(S)+ \sum_{i \in T} z_i}$.
Since $S, S'$ hash in the same bin,
$|\textnormal{profit}(S' ) - \textnormal{profit}(S)| \leq \epsilon K /2n$
and $|P_1(S)-P_1(S')| \leq \epsilon K /2n$.
Therefore, $|\textnormal{profit}(S' \cup T) - \textnormal{profit}(S \cup T)| \leq \epsilon K /n$.
\end{proof}


\noindent
{\bf Lemma \ref{lemma:induction}}.
For every $k =0, 1, \ldots, n$, after the $k^{th}$ iteration of the loop, there is a set $S$ in the hash table that
can be extended with elements from $\{ k+1, \ldots, n \}$ to a feasible set that has profit $\geq \textnormal{OPT}- \eps k \cdot K/n$.
\begin{proof}
By induction on $k$. The basis, $k=0$ of the claim is trivial: The hash table contains initially $\emptyset$, which can be extended to an optimal solution.

For the induction step, assume that the property holds after the
$(k-1)^{st}$ iteration for a set $S \subseteq [k-1]$ in the table, and let $T \subseteq \{ k, \ldots, n \}$ be an extension that yields a 
feasible set $S \cup T$ with profit within $\eps(k-1)\cdot K/n$ of OPT.
Suppose first that $k \notin T$. At the end of the $k^{th}$ iteration,
the hash table contains either $S$ or another set $S'$ that hashes in
the same bin and has replaced $S$, thus $S'$
has the same denominator but smaller numerator.
In the latter case, by Lemma~\ref{lemma:invariance}, $S' \cup T$ is also feasible,
and its profit is within $\epsilon K/n$ of the profit of $S \cup T$,
hence it is within $\eps k \cdot K/n$ of OPT.

The argument in the case $k \in T$ is similar. 
Since $S \cup T$ is feasible, $S \cup \{k\}$ is also feasible,
thus the algorithm will hash it and either insert it in the table
or not depending on whether there is another ``better'' set already
in the same bin. At the end of the $k^{th}$ iteration,
the hash table will contain either  the set $S \cup \{k\}$
or a set ${\hat S}$ that is at least as good
(has at least as low numerator) in the corresponding bin.
Whichever of these sets is in that bin at the end of the
iteration satisfies the property.
This is obvious for $S \cup \{k\}$, and it follows from Lemma~\ref{lemma:invariance} for ${\hat S}$:
Since $S \cup T = (S \cup \{k\}) \cup (T - \{k\})$ is feasible,
and $ {\hat S}$ hashes in the same bin as $S \cup \{k\}$ and
is at least as good, ${\hat S} \cup (T - \{k\})$ is also feasible,
it has profit within $\epsilon K/n$ of $S \cup T$, hence within $\eps k \cdot K/n$ of OPT.
\end{proof}

\subsection{Proof of Hardness}
We now give the proof that the flower PDP is NP-complete. 

\noindent
{\bf Theorem \ref{thm:hardness}}.
The PDP is NP-complete. 
\begin{proof}
We reduce from the Partition problem: Given a set of positive integers $a_1, \ldots, a_n$,
is there a partition of the numbers into two subsets that have equal sums?
We first apply the following (standard) transformation which yields an instance
of the partition problem where the numbers are comparable in value
and any solution must also have equal number of elements in each part.
Let $H = n \sum_i a_i$. Construct a new instance of the
partition problem with $2n$ elements $b_1, \ldots b_{2n}$; 
the first $n$ elements are $b_1=H+a_1, \ldots, b_n=H+a_n$,
and the other $n$ elements $b_j, j=n+1, \ldots, 2n$ are all $H$.
It is easy to see that the original instance has a solution iff the new instance does,
and furthermore, any solution of the new instance must have $n$ elements in each part.

We create now an instance of the PDP problem. The flower has $2n+1$ petals. 
The first $2n$ petals correspond to the $2n$ numbers $b_i$, and the last petal 
is the {\em special} petal.
We set the parameters as follows.
Set $\cL{i} =0$ for all $i$. 
Set $p_i = n^2(1-q_i)$ and $y_i = \frac{1}{n^2+1}(1-q_i)$ for all $i$.
Then $\lambda_i = \frac{p_i}{1-q_i}= n^2 $,
$w_i = \frac{p_i}{1-q_i-y_i} = n^2+1$,
and $z_i = \frac{p_i}{1-q_i-y_i} -\frac{p_i}{1-q_i} = 1$.
Therefore, $A= 0$ and $B= 1+ \sum_i \lambda_i = 1+n^2(2n+1)$.

We choose the platform coefficient for the special petal $s=2n+1$
so that its potential $\phi(s) = (\sum_i b_i)/2B = (2nH+ \sum_i a_i)/2B$.
Specifically, set $\cP{s} = (\sum_i b_i)/2B(n^2+1)$.
For the non-special petals $i \in [2n]$, we choose their platform coefficients so
that their potentials satisfy $\phi(i) = \phi(s) + b_i$.
For this, set $\cP{i} = ((\sum_i b_i)/2B +b_i)/(n^2+1)$.

Set $\textnormal{cost}_i=0$ for all $i$. Set $d_i = b_i$ for all $i \in [2n]$,
and for the special petal $s=2n+1$, we set $d_s = 4nH$.
This concludes the specification of the instance of PDP.

We shall show that the given instance of the Partition problem has a solution
if and only if the optimal profit is $v^*=(n^2+1)(4nH + \sum_{i=1}^{2n} b_i /2)/(B+n+1)$.
In this case, an optimal solution of the PDP instance consists of the
special petal and a solution of the partition instance.

First, suppose that the partition instance has a solution $S \subset [2n]$.
Consider the solution $S \cup {s}$ of the PDP instance.
We claim that it is feasible. The smallest potential is that
of the special petal, $\phi(s)$.
The utility of the agent for $S \cup \{s\}$ is 
$\frac{A+\sum_{i \in S \cup \{s\}} z_i \phi(i)}{B+\sum_{i \in S \cup \{s\}} z_i}$
which is equal to $\phi(s)$, since $\phi(i) = \phi(s) + b_i$ for all $i \in [2n]$
and $\sum_{i \in S} b_i = \sum_{i=1}^{2n} b_i /2$.\footnote{We assumed here for simplicity that the Agent's greedy algorithm includes a state in case of equality between the potential and the utility;
recall that from the Agent's perspective this does not make any difference. If the state is not included in case of equality, then
we have to adjust slightly the parameters to increase slightly the potential $\phi(s)$.}
An easy calculation also shows that the profit of the  solution $S \cup \{s\}$
is $(n^2+1)(4nH + (\sum_{i=1}^{2n} b_i /2))/(B+n+1)=v^*$.

Conversely, suppose that the PDP instance has a solution $S^*$ with profit
at least $v^*=(n^2+1)(4nH + \sum_{i=1}^{2n} b_i /2)/(B+n+1)$.
Then it must contain the special petal $s$, because even if we take all the
other petals, the profit is smaller.
Let $S= S^* - \{s\}$.
The profit of the solution $S^* = S \cup \{s\}$ is
$(n^2+1)(4nH + \sum_{i \in S} b_i )/(B+|S|+1)$.
If $|S| < n$, then $\sum_{i \in S} b_i \leq (n-1)H + \sum_i a_i$
and the profit is less than $v^*$.
Therefore $|S| \geq n$ and $\sum_{i \in S } b_i \geq \sum_{i=1}^{2n} b_i /2$.
Since $S^*$ is feasible, we must have $\phi(s) \geq \frac{A+\sum_{i \in S } z_i \phi(i)}{B+\sum_{i \in S } z_i}$.
Substituting the values of the parameters, this inequality yields,
$\sum_{i \in S } b_i \leq \sum_{i=1}^{2n} b_i /2$.
Therefore, for the profit to be $v^*$, we must have $|S|=n$
and $\sum_{i \in S } b_i = \sum_{i=1}^{2n} b_i /2$.
Thus, the partition instance has a solution.
\end{proof}

%% file: appendixC.tex
\section{Proofs for Section~\ref{sec.manyAgents}}
\label{appendixC} 
\subsection{Proofs of Lemmas used in Theorem~\ref{thm:multi-Agent-PDP}}

\noindent
{\bf Lemma \ref{lem.multi1}}.
For every pair $(\theta,D) \in (\Phi, {\cal D})$,
if a set $S$ hashes into a slot $(a,b)$ that is consistent
with $(\theta,D)$, 
then $\textnormal{value}_{(\theta,D)}(S) = \textnormal{profit}(S)$.
In particular, the set $S(\theta,D)$ selected by the algorithm (if any) 
satisfies $\textnormal{value}_{(\theta,D)}(S(\theta,D)) = \textnormal{profit} ( S(\theta,D))$.
\begin{proof}
Let $(a,b)= \textnormal{hash}(S)$.
The slot is consistent with $(\theta,D)$,
thus, $D_i = B_i + b_i \delta$, and 
$\theta_i > \frac{A_i + a_i \delta \delta'}{D_i} \geq \theta'_i$ for all $i \in [k]$.
Since $S$ hashes into slot $(a,b)$, we have
$a_i = \sum_{j \in S \cap Q_i(\theta_i)} \frac{z_{ij} \phi_{ij}}{\delta \delta'}$ and $b_i = \sum_{j \in S \cap Q_i(\theta_i)} \frac{z_{ij}}{\delta}$ for all $i \in [k]$.

For each Agent $i$, consider the operation of the greedy algorithm
when offered the set of platforms $S$.
Since $\theta_i > \frac{A_i + a_i \delta \delta'}{D_i} \geq \theta'_i$,
the greedy algorithm will select precisely all states 
$j$ of $S$ that have potential $\geq \theta_i$,
i.e. $\textnormal{Agent}_i(S) = S \cap Q_i(\theta_i)$.
Therefore, $\textnormal{profit}(S) = \textnormal{value}_{(\theta,D)}(S)$.
\end{proof}

\noindent
{\bf Lemma \ref{lem.multi2}}.
Let $S^*$ be an optimal solution to the Platform Design Problem,
and let $\theta_i = \min_{j \in \textnormal{Agent}_i(S^*)} \{ \phi_i(j)\}$, $D_i = B_i + \sum_{j \in \textnormal{Agent}_i(S^*)} z_{ij}$.
Then, in the iteration for the pair $(\theta,D)$, the algorithm
selects a set $S(\theta,D)$, and the set has
$\textnormal{profit}(S(\theta,D)) \geq \textnormal{profit}(S^*)$.
\begin{proof}
Consider the iteration of the algorithm for the pair $(\theta,D)$.
We can show by induction on $t =0,1, \ldots, n$,
that after stage $t$, the slot $\textnormal{hash}(S^* \cap [t])$
is nonempty, and the value of the set in the slot is at least
the value of $S^* \cap [t]$.

The claim is obviously true in the beginning.
Consider stage $t$. If $t \notin S^*$, then
the slot $\textnormal{hash}(S^* \cap [t-1]) =\textnormal{hash}(S^* \cap [t])$,
and the set in this slot at the end of stage $t$ is either the same
as the set after stage $t-1$, or another set with higher value;
thus the claim follows from the induction hypothesis.

Suppose $t \in S^*$, let $(a,b) = \textnormal{hash}(S^* \cap [t-1])$,
and let $S= H(a,b)$. By the induction hypothesis,
$\textnormal{value}_{(\theta,D)}(S) \geq \textnormal{value}_{(\theta,D)}(S^* \cap [t-1])$.
Then
$S \cup \{t\}$ hashes in the same slot as $S^* \cap [t]$,
and $\textnormal{value}_{(\theta,D)}(S \cup \{t\})
= \textnormal{value}_{(\theta,D)}(S) + c_t(\theta,D) \geq 
\textnormal{value}_{(\theta,D)}(S^* \cap [t-1]) + c_t(\theta,D) = \textnormal{value}_{(\theta,D)}(S^* \cap [t])$.
At the end of stage $t$, this slot has either the set $S \cup \{t\}$
or another set with a higher value. 

After stage $n$, the set $S^*$ hashes into a slot $(a,b)$
that is consistent with $(\theta,D)$ from our choice of
$\theta$ and $D$. By the claim above,
this slot has a set that has equal or larger value.
By Lemma \ref{lem.multi1}, the value is equal to the profit.
Therefore, $S(\theta,D)$ exists and
$\textnormal{profit}(S(\theta,D) \geq \textnormal{profit}(S^*)$. 
\end{proof}

\subsection{Proof of Hardness}

\noindent
{\bf Theorem \ref{thm:hard2}}.
Unless P=NP, there is no FPTAS for the Designer's problem with
two agents if the $\phi_i(j)$ are not restricted to be polynomially bounded.

\begin{proof}
We reduce from the Partition problem: Given a set of positive integers $a_1, \ldots, a_n$,
is there a partition of the numbers into two subsets that have equal sums?
As in the proof of Theorem \ref{thm:hardness},
we first transform it to an instance
of the partition problem with $2n$ numbers 
$b_1=H+a_1, \ldots, b_n=H+a_n$,
and  $b_j =H$ for $j=n+1, \ldots, 2n$, where $H = n \sum_i a_i$. 
The original instance has a solution iff the new instance does,
and furthermore, any solution of the new instance must have $n$ elements in each part.

We have two agents. The agents have a flower Markov chain
with $2n+1$ petals, where the first $2n$ petals correspond to the $2n$ numbers $b_i$, and the last petal 
is the {\em special} petal.
The parameters of the Markov chain for both agents are the same as in the proof of Theorem \ref{thm:hardness}.
That is, for $i=1,2$ we set 
$\cL{ij} =0$ for all $j$; 
$p_{ij} = n^2(1-q_{ij})$ and $y_{ij} = \frac{1}{n^2+1}(1-q_{ij})$ for all $j$.
Then $\lambda_{ij} = \frac{p_{ij}}{1-q_{ij}}= n^2 $,
$w_{ij} = \frac{p_{ij}}{1-q_{ij}-y_{ij}} = n^2+1$,
and $z_{ij} = \frac{p_{ij}}{1-q_{ij}-y_{ij}} -\frac{p_{ij}}{1-q_{ij}} = 1$.
Therefore, $A_1=A_2= 0$ and $B_1=B_2=B= 1+ \sum_j \lambda_{ij} = 1+n^2(2n+1)$.

The two agents differ in the rewards when they adopt the platform
in a state.
Agent $1$ has the same rewards as in the proof of Theorem \ref{thm:hardness}. 
Agent $2$ has rewards that are defined in a similar way from the 
numbers $b'_j = 2H-b_j$ for all $j \in [2n]$.
Thus, for the special petal $s=2n+1$, we set
$\cP{1s} = (\sum_i b_i)/2B(n^2+1)$ and 
$\cP{2s} = (\sum_i b'_i)/2B(n^2+1)$.
Therefore, its potential for the two agents is
$\phi_1(s) = (\sum_i b_i)/2B$
and $\phi_2(s) = (\sum_i b'_i)/2B$.
For the non-special petals $i \in [2n]$, we choose their platform coefficients so
that their potentials satisfy $\phi_1(j) = \phi_1(s) + b_i$
and $\phi_2(j) = \phi_2(s) + b'_i$.
For this, set $\cP{1j} = ((\sum_i b_i)/2B +b_j)/(n^2+1)$,
and $\cP{2i} = ((\sum_i b'_i)/2B +b'_j)/(n^2+1)$.
Set $\textnormal{cost}_j=0$ for all $j$. 
Set $d_{1j} = d_{2j} =1$ for all $j \in [2n]$,
and for the special petal $s=2n+1$, we set $d_{1s}=d_{2s} = 3n$.
This concludes the specification of the instance of PDP.

We shall show that if the given instance of the Partition problem has a solution then 
the optimal profit is $v^*=8n\cdot \frac{n^2+1}{B+n+1}$,
whereas if it does not have a solution
then the optimal profit is at most 
$(8n-2)\cdot \frac{n^2+1}{B+n} < (1-\frac{1}{8n})v^*$.

First, suppose that the partition instance has a solution $S \subset [2n]$.
Consider the solution $S \cup \{s\}$ of the PDP instance.
It is easy to check that both agents will adopt all the states in
$S \cup \{s\}$, as in the proof of Theorem \ref{thm:hardness}. 
The profit of the solution $S \cup \{s\}$
is $8n\cdot \frac{n^2+1}{B+n+1}$.

Conversely, suppose that the PDP instance has a solution $S^*$ with profit
greater than $(8n-2)\cdot \frac{n^2+1}{B+n}$.
Then $S^*$ must contain the special petal $s$,
and $s$ must be selected by both agents, because otherwise,
even if they take all the other petals, the profit is no more
than $7n\cdot \frac{n^2+1}{B+n+1}$.
Since $s$ has lowest potential among all the petals for both agents,
it follows that both agents select all the states in $S^*$.

Let $S= S^* - \{s\}$.
If $|S| \geq n+1$ then Agent $1$ will not select the special state because
$\phi_1(s) < \frac{A_1+ \sum_{j \in S} z_{1j} \phi_1(j)}{B_1+ \sum_{j \in S} z_{1j}} = \frac{\sum_{j \in S} (b_j + \phi_1(s))}{B_1+|S|}$
since $\sum_{j \in S} b_j \geq |S| \cdot H \geq (n+1)H$,
and $B_1 \phi_1(s) = (\sum_{i \in [2n]} b_i)/2 < (n+1)H$.
Therefore, $|S| \leq n$.
On the other hand, if $|S| \leq n-1$ then 
$\textnormal{profit}(S^*) \leq 2(4n-1)\cdot \frac{n^2+1}{B+n}$.
Therefore, $|S| =n$.

Since both agents select the special petal, we have
from Agent $1$:
$\phi_1(s) \geq \frac{A_1+ \sum_{j \in S} z_{1j} \phi_1(j)}{B_1+ \sum_{j \in S} z_{1j}} = \frac{\sum_{j \in S} b_j + n \phi_1(s)}{B+n}$,
therefore, $\phi_1(s) \geq \frac{\sum_{j \in S} b_j }{B_1}$.
Since $\phi_1(s)= \frac{\sum_{i \in [2n]} b_i}{2B_1}$,
we get $\sum_{j \in S} b_j \leq \frac{\sum_{i \in [2n]} b_i}{2}$.
Similarly, we get from Agent $2$:
$\sum_{j \in S} b'_j \leq \frac{\sum_{i \in [2n]} b'_i}{2}$.
Since $b'_j = 2H-b_j$ and $|S|=n$,  this implies that
$\sum_{j \in S} b_j \geq \frac{\sum_{i \in [2n]} b_i}{2}$.
Therefore, $\sum_{j \in S} b_j = \frac{\sum_{i \in [2n]} b_i}{2}$,
and $S$ is a solution to the Partition problem.
\end{proof}

%% file: appendixD.tex


\input{agent-multiplatform} \label{appendixD}

%% file: agent-multiplatform.tex
\section{The Agent's Problem with Multiple Platforms per State} \label{sec:agent-multiplatform}

We consider now the flower MDP setting where there can be multiple available platforms for the same state. This is the case for example when there are multiple designers that offer a platform for the same state.
The agent will select for each state either one of the available platforms or no platform.

We are given a set of available platforms, where each platform
is associated with one state of the flower MDP.
For each available platform $j$, we are given the associated agent's reward and the change in the transition probabilities of the state;
these induce the corresponding parameters $z_j$ and $\phi(j)$ as
in Section \ref{sec:agent}.
The agent will select a subset $S$ of platforms that contains at most
one platform for each state; call such a set `feasible'.
The agent's utility $u(S)$ for a feasible set $S$, is given by
formula \ref{eq:agent-opt} in Section \ref{sec:agent}: 
$u(S)= \frac{A+\sum_{j \in S}z_j \phi(j) }{B + \sum_{j \in S} z_j}$.
The agent's objective is to select a feasible set $S$ that maximizes $u(S)$.

We observe first that the following straightforward extension of the greedy algorithm does not work: Sort the platforms in decreasing order of potential $\phi(j)$ and initialize $S$ to $\emptyset$. For each platform $j$ in this order, if $\phi(j)$ exceeds the utility of the current solution set $S$ and $S$ does not contain a platform for the same state, then add $j$ to $S$ else discard it.
This algorithm can produce a suboptimal solution.
For example, suppose there is one leaf state with two possible platforms,
platform 1 has $z_1=1$, $\phi(1) = 5$, platform 2 has $z_2=2, \phi(2)= 4$,
and $A=B=10$ in the objective function.
If we select platform 1, the utility is 15/11, while if we select 2 it is $18/12 > 15/11$.

Consider the following three possible changes to a feasible solution $S$:
(1) Remove a platform from $S$, (2) Add a platform to $S$ if $S$ does not contain another platform for the same state, (3) Swap a member of $S$ with another platform for the same state that is not in $S$.
We show first that if a feasible solution is locally optimal under these types of changes, i.e. cannot improved, then it is globally optimal. 

The first type of move that removes a platform $j \in S$ improves the utility
if $\phi(j) < u(S)$; the second type of move that adds $j$ improves the utility if $\phi(j) > u(S)$. 
The following lemma states when a swap increases the utility.

\begin{lemma}\label{lem:swap}
Let $S$ be any feasible set, let $j, j'$ be two platforms for the same state where $j \in S$, $j' \notin S$, and let $S' = S \cup \{j'\} \setminus \{j\}$.\\
1. If $z_j = z_{j'}$ then $u(S) < u(S')$ iff $\phi(j) < \phi(j')$.\\
2. 
Let $\rho(j,j') = \frac{z_{j'}\phi(j')-z_j \phi(j)}{z_{j'} - z_j}$.
If $z_j < z_{j'}$ then $u(S')$ lies between $\rho(j,j')$ and $u(S)$, i.e., either $\rho(j,j') < u(S') < u(S)$
or $\rho(j,j') = u(S') = u(S)$
or $\rho(j,j') > u(S') > u(S)$.
If $z_j > z_{j'}$ then $u(S)$ lies between $\rho(j,j')$ and $u(S')$.
\end{lemma}
\begin{proof}
$u(S)$ has the form
$\frac{{\tilde A} + z_j\phi(j)}{{\tilde B}+z_j}$
where ${\tilde A}, {\tilde B}$ include the contributions of all the members of $S\setminus\{j\}$,
and $u(S') = \frac{{\tilde A} + z_{j'}\phi(j')}{{\tilde B}+z_{j'}}$.
If $z_j= z_{j'}$ then the statement of the lemma is obvious.

Suppose that $z_j < z_{j'}$.
The utility $u(S')$ can be obtained from $u(S)$
by adding $z_{j'} \phi(j')-z_j \phi(j)$ to the numerator and
$z_{j'} -z_j$ to the denominator.
The statement follows from Lemma \ref{lemma:1}.
If $z_j > z_{j'}$ then the proof is the same: switch the roles of $j$ and $j'$ (note that $\rho(j,j')= \rho(j',j)$).
\end{proof}

\begin{lemma}\label{lem:localopt}
A feasible set $S$ is optimal if and only if 
(1) for every $j \in S$, $\phi(j) \geq u(S)$, and (2) for every 
$j' \notin S$, either $\phi(j') \leq u(S)$ or $S$ contains another
platform $j$ for the same state and swapping $j$ for $j'$ does not
increase the utility.
\end{lemma}
\begin{proof}
The one direction is obvious: If $S$ contains a platform $j$ with $\phi(j) <u(S)$ then removing $j$ from $S$ increases the utility.
If there is a $j' \notin S$ such that $\phi(j')> u(S)$ and $S$ does not
contain any platform for the same state, then adding $j'$ to
$S$ increases the utility. Finally if $S$ contains a platform $j$ for the
same state and swapping $j$ with $j'$ increases the utility then $S$ is not optimal.

For the other direction, suppose that $S$ satisfies the conditions of the lemma and is not optimal. Let $S' = S \cup X' \setminus Y'$
be an optimal solution where $S \cap X' = \emptyset$ and $Y' \subseteq S$.
Since $S'$ is optimal, $\phi(j') \geq u(S')$ for all $j' \in X'$.
For every $j' \in X'$, since $j' \notin S$ and $\phi(j') \geq u(S') > u(S)$,
there is another $j \in S$ for the same state 
(hence $j$ must be in $Y'$ since $S'$ is feasible) and swapping $j$ for $j'$ does not increase the utility.
Let $X = \{ j | j' \in X' \}$ and let $Y= Y' \setminus X$.
For any set $Q$ of platforms, let $f(Q) =  \sum_{j \in Q} z_{j} \phi(j)$
and $g(Q)= \sum_{j \in Q} z_{j} $.
Then $u(S)$ can be written as 
$u(S) = \frac{A + f(S\cap S') + f(X)+f(Y)}{B +g(S\cap S')+g(X)+g(Y)}$,
and $u(S') = \frac{A + f(S\cap S')+f(X')}{ B +g(S\cap S') +g(X')}$.
Since $\phi(j) \geq u(S)$ for all $j \in S$,
it follows that $\frac{f(Y)}{g(Y)} \geq u(S)$, hence
$\frac{A +f(S\cap S')+ f(X)}{B+g(S\cap S') +g(X)} \leq u(S)$.

Let $X_0 = \{j \in X | z_j = z_{j'} \}$, 
$X_1 = \{ j \in X | z_j < z_{j'} \}$
and $X_2= \{ j \in X | z_j > z_{j'} \}$.
Since swapping each $j \in X$ in $S$ 
for the corresponding $j' \in X'$ does not
increase the utility, we have $f(j') \leq f(j)$ for every $j \in X_0$,
$\frac{f(j') -f(j)}{g(j')-g(j)} \leq u(S)$ for every $j \in X_1$
and  $\frac{f(j') -f(j)}{g(j')-g(j)} \geq u(S)$ for every $j \in X_2$.
Therefore, $f(X'_0) \leq f(X_0)$ (while $g(X'_0) = g(X_0)$);
$\frac{f(X_1') -f(X_1)}{g(X_1')-g(X_1)} \leq u(S)$ (the denominator here is positive);
and $\frac{f(X_1') -f(X_1)}{g(X_1')-g(X_1)} \geq u(S)$ (the denominator here is negative).
Since $\frac{A +f(S\cap S')+ f(X)}{B+g(S\cap S') +g(X)} \leq u(S)$,
it follows that $u(S') = \frac{A +f(S\cap S')+ f(X')}{B+g(S\cap S') +g(X')}\leq u(S)$.
\end{proof}

Consider the platforms associated with the same state.
We can eliminate platforms in the problem instance
that are dominated by
other platforms and thus are not
needed to attain the optimal utility. The following lemma gives several types of dominated platforms.

\begin{lemma}\label{lem:dom}
If a platform $j$ satisfies one of the following properties,
then it can be removed from the instance without changing the optimal utility.\\
1. There is another platform $j'$ for the same state such that
$z_j \leq z_{j'}$ and $\phi(j) \leq \phi(j')$, with at least one of the inequalities strict.\\
2. There is another platform $j'$ for the same state such that
$z_j > z_{j'}$ and $z_j \phi(j) \leq z_{j'} \phi(j')$.\\
3.  There are two platforms $k, l$ for the same state such that
$z_k < z_j < z_{l}$ and $\rho(k,j) < \rho(j,l)$.
\end{lemma}
\begin{proof}
Assume that $j$ belongs to an optimal feasible solution $S$.

1. We claim that swapping $j$ for $j'$ will produce a solution $S'$ with the same or higher utility.
If $z_j \leq z_{j'}$ and $\phi(j) < \phi(j')$ then obviously $u(S')>u(S)$.

Assume that $z_j < z_{j'}$ and $\phi(j) \leq \phi(j')$.
Since $S$ is optimal, $u(S) \leq \phi(j)$.
Note that $\phi(j) = \frac{z_{j}\phi(j)}{z_{j}}$, and
$\phi(j') = \frac{z_{j'}\phi(j')}{z_{j'}}$=
$\frac{z_{j}\phi(j)+ (z_{j'}\phi(j') - z_{j}\phi(j))}{z_{j'}-z_j}$.
Thus, by Lemma \ref{lemma:1}, $\phi(j')$ lies between
$\phi(j)$ and $\rho(j,j') = \frac{z_{j'}\phi(j') - z_{j}\phi(j)}{z_{j'}-z_j}$. Since $\phi(j) \leq \phi(j')$, we have
$u(S) \leq \phi(j) \leq \phi(j') \leq \rho(j,j')$.
Lemma \ref{lem:swap} implies then that $u(S) \leq u(S')$.
The inequality is strict, unless $u(S) = \phi(j) = \phi(j')$,
which means that all platforms in $S$ have the same potential,
equal to $u(S)$, and also the fixed constants $A, B$ satisfy $\frac{A}{B}=u(S)$.

2. If $z_j \phi(j) \leq z_{j'} \phi(j')$ then
$\rho(j',j) \leq 0$, thus $\rho(j',j) < u(S)$, and by Lemma \ref{lem:swap},  $u(S) < u(S')$.

3. Since $S$ is optimal, swapping $j$ for $k$ or $l$ does not
increase the utility. Therefore, by Lemma \ref{lem:swap},
$\rho(k,j) \geq u(S)$ and $\rho(j,l) \leq u(S)$,
hence,  $\rho(k,j) \geq \rho(j,l)$.
\end{proof}

In a preliminary step we can process separately for each state the associated platforms and eliminate those that are redundant, i.e. satisfy one of the conditions of Lemma \ref{lem:dom}.
Platforms that have the same $z$ and $\phi$ value can be identified
(they are indistinguishable as far as the agent is concerned).
Let $j_1, j_2, \ldots, j_k$ be the nonredundant platforms for a state
in decreasing order of potential. By condition 1 of Lemma \ref{lem:dom}, they increase in
$z$ value: $z_{j_1} < z_{j_2} < \ldots < z_{j_k}$.
By condition 2 they also increase in the value of $z \cdot \phi$:
$z_{j_1} \phi(j_1) < z_{j_2} \phi(j_2)  < \ldots < z_{j_k} \phi(j_k)$.
Map every platform $j$ of the state to a point 
$p_j= (z_j, z_j \phi(j))$ on the plane.
Note that the ratio $\rho(j,j')$ for two platforms $j, j'$ is the
slope of segment $(p_j,p_{j'})$. Condition 3 of Lemma \ref{lem:dom}
says that if $p_j$ is below the segment $(p_k, p_l)$ of two other points then $j$ is redundant. These conditions imply that the
nonredundant platforms correspond to the points that lie on the
convex Pareto curve of the point set, 
i.e. the upper envelope of the convex hull of the
collection of points for the platforms of the state.
This curve $p(j_1), p(j_2), \ldots, p(j_k)$ is a piecewise linear
increasing concave curve; all the slopes are positive and decreasing.
For every nonredundant platform $j$ we use $prev(j)$ to denote the
previous platform in the sequence for its state (if it exists, i.e.
$prev(j_i) = j_{i-1}$ if $i>1$), and $next(j)$ the next platform (if $j<k$).

To compute the nonredundant platforms, we first sort all the platforms in decreasing order of potential $\phi$, with ties broken by $z$ value (smallest first), and then process separately for each state its platforms in order; a simple linear scan suffices to remove the
redundant platforms.

An optimal feasible solution contains at most one platform for each state. Suppose that the optimal utility is $u^*$. 
If we know $u^*$, then we can easily construct an optimal solution.
We have the following criterion for optimality of a solution $S$
based on its utility $u(S)$.

{\bf Lemma \ref{lem:multi-criterion}}.
Let $S$ be a feasible set of nonredundant platforsm.
The set $S$ is optimal if and only if for every states $s$, either
(1) $S$ does not contain any platform for $s$ and all platforms $j$ for $s$ have potential $\phi(j) \leq u(S)$, or
(2) the platform $j \in S$ for state $s$ satisfies 
(i) $\rho(prev(j),j) \geq u(S)$ if $prev(j)$ exists, else $\phi(j) \geq u(S)$, and (ii) $\rho(next(j),j) \leq u(S)$ if $next(j)$ exists.
\begin{proof}
Suppose that $S$ satisfies the conditions of the lemma.
Use Lemma \ref{lem:localopt}. 
Consider $j \in S$ and let $s$ be its state.
Either $j$ is the first platform in the Pareto curve for $s$,
in which case $\phi(j) \geq u(S)$, or else $\rho(prev(j),j) \geq u(S)$. Since $phi(j)$ lies between $\phi(prev(j)$ and $\rho(prev(j),j)$,
and since $\phi(prev(j) > \phi(j)$, it follows that $ \phi(j) > \rho(prev(j),j) \geq u(S)$.

Let $j'$ be any other platform for the same state $s$.
If $z_{j'} > z_j$ then $\rho(j,j') \leq \rho(j,next(j) \leq u(S)$
(or $next(j)$ does not exist which means that $z_{j'} \phi(j') \leq z_{j} \phi(j)$ and $j'$ is redundant).
If $z_{j'} < z_j$ then $\rho(j',j) \geq \rho(prev(j),j) \geq u(S)$
(or $prev(j)$ does not exist and $j'$ is again redundant.
If $z_{j'} = z_j$ then $\phi(j') \leq \phi(j)$.
In all cases, swapping $j$ for $j'$ does not increase the utility.
It follows from Lemma \ref{lem:localopt} that $S$ is optimal.

For the other direction, note that
if $S$ does not satisfy condition (1) then we can increase its utility
by adding a platform $j$ for state $s$ with potential $\phi(j) > u(S)$. If $S$ does not satisfy condition (2), then we can increase
the utility by swapping $j$ for $prev(j)$ or $next(j)$, or by removing $j$ if $j$ is the first platform and $\phi(j) < u(S)$. 
\end{proof}

Of course we do not know ahead of time the optimal utility $u^*$.
We will compute $u^*$ and an optimal solution using a greedy algorithm
with a different parameter $\psi(j)$ for each nonredundant platform.
If $j$ is the first nonredundant platform in the sequence 
for its state, then
set $\psi(j) = \phi(j)$, otherwise set $\psi(j)= \rho(prev(j),j)$.
Note that $\psi(j) \leq \phi(j)$ for all $j$.
To see this for platforms $j$ other than the first one 
in the sequence for its state, observe that 
$\phi(j)$ lies between $\phi(prev(j)$ and $\rho(prev(j),j)$.
Since $\phi(prev(j) > \phi(j)$ it follows that
$\phi(prev(j) > \phi(j) > \rho(prev(j),j) =\psi(j)$.

The algorithm is given in Algorithm \ref{alg:agent-multi} of
Section \ref{sec:competitive}.
As we'll see, whenever we add a platform $j$ to the solution $S$,
if $j$ is not the first platform for its state, then $S$ contained
previously $prev(j)$ and thus we remove it.

\smallskip
\noindent
{\bf Theorem \ref{thm:agent-multi-opt}}
Algorithm~\ref{alg:agent-multi} returns an optimal feasible solution.
The algorithm runs in time $O(n+m \log m)$, where
$n$ is the number of states and $m$ is the number of platforms.
\begin{proof}
By Lemma \ref{lem:dom}, removing the redundant platforms
does not change the optimal utility.
From the definition of $\psi$, for each state $s$,
the parameters are decreasing along the Pareto curve of the state.
A simple inductive argument shows that the utility $u(S)$ is
increasing in every iteration. 
This is clear if the platform $j$ of the iteration is the
first one for its state. If $j$ is not the first state, then
the last platform for this state that was processed was $prev(j)$,
and since at that time the utility was lower (by induction hypothesis),
and since $\psi(prev(j) > \psi(j) > u(S)$, we added $prev(j)$
at that iteration. The algorithm swaps $prev(j)$ for $j$,
and thus increases the numerator of $u(S)$ by the numerator of
$\rho(prev(j),j)$ and the denominator of $u(S)$ by the denominator of
$\rho(prev(j),j)$. Since $\psi(j) = \rho(prev(j),j)> u(S)$,
the utility increases.

When the algorithm stops and returns $S$, the solution $S$ satisfies
the criterion of Lemma \ref{lem:multi-criterion}, hence it is optimal.
\end{proof}

%% file: appendixE.tex
\input{manyplatforms} \label{appendixE}

%% file: manyplatforms.tex
\section{The Designer Problem in a Competitive Setting} \label{sec.manyplatforms}

We consider in this section the problem of a
designer choosing which platforms to offer when there
are already in the market available platforms from other
providers. We extend the algorithm of Section \ref{sec.manyAgents}
to this setting.
We have $k$ Agents, each with their own flower Markov chain
on the same state set (but different transition probabilities).
There is a set of existing available platforms (offered by other providers).
The Designer can build a platform for each state, and wants to select an optimal subset of platforms that maximizes the profit.

For every Agent $i$ and platform $j$ (both the existing and the Designer's potential platforms)
we have the corresponding parameters for the Agent's reward, and the
transition probabilities of the Markov chain; these induce 
as before corresponding values $z_{ij}$ and $\phi_i(j)$.
Each agent will adopt at most one platform for each state, to maximize
her utility.
The utility of Agent $i$ for a set $R$ of platforms is  
$u_i(R) = \frac{A_i + \sum_{j \in R} z_{ij} \phi_i(j)}{B_i + \sum_{j \in R} z_{ij}}$.
The profit function of the Designer if he offers the set $S$ of platforms is
\[
\textnormal{profit}(S) := \sum_i \frac{ \sum_{j \in \textnormal{Agent}_i(S)\cap S} d_{ij}\cdot\frac{p_{ij}}{1 - q_{ij} - y_{ij}}}{B_i + \sum_{l \in \textnormal{Agent}_i(S)} z_{il}} - \sum_{j \in S} \textnormal{cost}_j
\]
where $d_{ij}$ is the designer's reward rate if
Agent $i$ adopts platform $j \in S$, $\textnormal{cost}_j$ is the cost
of building platform $j$, and  $\textnormal{Agent}_i(S)$ is the set of platforms that Agent $i$ adopts (by any provider) if the designer offers the set $S$.

We assume again that the parameters $z_{ij}$ and $\phi_i(j)$ are quantized.
That is, we assume that  
each $z_{ij}$ is of the form $l_{ij} \delta$ for some integer
$l_{ij} \leq M$ and some $\delta$, with $M$ polynomially bounded, and similarly
each  $\phi_i(j)$ is of the form $l'_{ij} \delta'$ for some 
integer $l'_{ij} \leq M$ and some $\delta'$. 
We will show how to extend the algorithm of Section \ref{sec.manyAgents} to compute the optimal
solution for the Designer's problem in polynomial time for a fixed number of agents.

For each Agent $i$ we define a set $\Phi_i$ of critical values
as follows.  Consider a state $s$ and recall from Section \ref{sec:agent-multiplatform} the convex Pareto curve
formed by the nonredundant platforms for the state.
Denote by $P_i(s)$ the curve for the existing available platforms
for state $s$.
The set $\Phi$ includes the potential of the first
member of the Pareto curve $P_i(s)$ and the slopes of all the segments of the curve. 
If the Designer offers a platform for this state, we obtain 
a possibly modified Pareto curve $P'_i(s)$ that contains the designer's platform (if it is not redundant for agent $i$) and
possibly does not contain a subsequence of points of the previous curve. We include in the set $\Phi_i$ also the potential of the first platform and the slopes of the modified curve $P'_i(s)$ for the state.
The set $\Phi_i$ contains $\infty$
and the above sets of values for each state. 
Clearly, the size of $\Phi_i$ is linear in the number of platforms.

Let ${\cal D}_i = \{ B_i + l \delta | l \in [nM] \}$, 
${\cal N}_i = \{ A_i + l \delta \delta' | l \in [nM^2] \}$.
Note that $|{\cal D}_i|, |{\cal N}_i|$ are polynomially bounded
by our assumption. By the definitions,
for every subset $R$ of platforms adopted by agent $i$, 
the numerator of
the utility $u_i(R)$ is in ${\cal N}_i$ and the denominator is in ${\cal D}_i$. 
Let $\Phi = \Pi_{i=1}^k \Phi_i$,
${\cal D} = \Pi_{i=1}^k {\cal D}_i$, 
and ${\cal N} = \Pi_{i=1}^k {\cal N}_i$.

Recall from Section \ref{sec:agent-multiplatform} that
in the Agent's problem, if we know the optimal utility $u^*$
we can easily determine an optimal solution:
if the Pareto curve for a state $s$ is $j_1, j_2, \ldots, j_r$,
then an optimal solution contains some platform for state $s$
if $\phi_i(j_1) > u^*$, and then in particular it contains the
platform $j_l$ such that the slope of the previous segment
$(j_{l-1},j_l)$ of the curve (if $l>1$) is $> u^*$ and the slope
of the next segment $(j_{l},j_{l+1})$ (if $l<r$) is $\leq u^*$.
Clearly, we do not need to know precisely the value of $u^*$
to make this determination: it suffices to know how $u^*$
compares with the elements of $\Phi_i$.

For any $\theta_i \in \Phi_i$, let $\theta'_i$ be the next smaller value in $\Phi_i$ (if $\theta_i$ is the minimum of $\Phi_i$ then 
let $\theta'_i = -1$). 
Define $Q_i(\theta_i)$ to be the set of all platforms $j$ of the Designer such that $j$ belongs to the (new) Pareto curve $P'_i(s)$ for the state $s$ corresponding to $j$, 
the segment of the curve before $j$ has slope $\geq \theta_i$
or $j$ is the first platform of the curve and it has potential $\phi_i(j) \geq \theta_i$,  and the next segment after $j$ (if it exists) has slope $\leq \theta'_i$.
Thus, if the Designer offers platform $j$ (along with some other
subset of platforms) and the Agent's optimal utility $u^*$ satisfies
$\theta'_i \leq u^* < \theta_i$ then the optimal solution
includes platform $j$.

For any tuple $\theta \in \Phi$ and tuple $D \in {\cal D}$
and platform $j$ of the Designer,
define a corresponding value coefficient
$$c_j(\theta,D) = \sum_{i: j \in Q_i(\theta_i)}  \frac{d_{ij} p_{ij}}{(1-q_{ij}-y_{ij})D_i }  -\textnormal{cost}_j$$

The summation in the above formula includes only
those $i \in [k]$ such that $j \in Q_i(\theta_i)$.
For any subset $S$ of platforms of the Designer,
define $\textnormal{value}_{(\theta,D)}(S) = \sum_{j \in S} c_j(\theta,D)$.

The algorithm is formally the same as Algorithm 3 of Section \ref{sec.manyAgents} with the difference that
we use the above definition of the sets $Q_i(\theta_i)$
and modify also the definition of the hashing function
and the notion of consistency.
Recall that for every tuple $\theta \in \Phi$ and  $D \in {\cal D}$,
the algorithm processes the states
and associated platforms of the Designer in arbitrary order $1, \ldots, n$, and it
employs a hash table $H$ indexed by two $k$-tuples $a, b$
of integers, where $a \in ([M^2])^k$, $b \in M^k$.
Fix a tuple $\theta, D$.
For each state $s$, let $f_i(s) \in P_i(s)$ be the platform from state $s$ that Agent $i$ adopts (if any) in case that the Designer does not offer a platform for state $s$ and the Agent's optimal utility lies in
the interval $(\theta'_i,\theta_i]$; in case the agent does not adopt
any platform from $s$, we let $f_i(s)$ be a dummy platform
with zero $z$ value and potential $\phi$.
For each possible platform $j$ of the Designer, 
if $s$ is the state of the platform, we
define $\sigma_i(j) = z_{ij} \phi_i(j) - z_{i f_i(s)} \phi_i(f_i(s))$
and $\tau_i(j) = z_{ij}  - z_{i f_i(s)}.$
We define the hashing function as follows.
A subset $S$ of platforms of the Designer hashes into
the slot $(a,b)$ where 
$a_i = \sum_{j \in S \cap Q_i(\theta_i)} \frac{\sigma_i(j)}{\delta \delta'}$ and 
$b_i = \sum_{j \in S \cap Q_i(\theta_i)} \frac{\tau_i(j)}{\delta}$
for all Agents $i \in [k]$.

Define ${\hat A}_i = A_i+ \sum_s z_{i f_i(s)} \phi_i(f_i(s))$
and ${\hat B}_i = B_i+ \sum_s z_{i f_i(s)}$.
We say that a slot $(a,b)$ is {\em consistent} with
the pair $(\theta,D)$ if
$D_i = {\hat B}_i +b_i\delta$ and 
$\theta_i > \frac{{\hat A}_i + a_i \delta \delta'}{D_i} \geq \theta'_i$ for all $i \in [k]$.

We have again the following property as in Lemma \ref{lem.multi1}
of Section \ref{sec.manyAgents}.

\begin{lemma} \label{lem.multides1}
For every pair $(\theta,D) \in (\Phi, {\cal D})$,
if a set $S$ hashes into a slot $(a,b)$ that is consistent
with $(\theta,D)$, 
then $\textnormal{value}_{(\theta,D)}(S) = \textnormal{profit}(S)$.
In particular, the set $S(\theta,D)$ selected by the algorithm (if any) 
satisfies $\textnormal{value}_{(\theta,D)}(S(\theta,D)) = \textnormal{profit} ( S(\theta,D))$.
\end{lemma}
\begin{proof}
Suppose that the Designer offers the set $S$ of platforms.
We claim that Agent $i$ will adopt the set  $R_i = (S \cap Q_i(\theta_i)) \cup \{ f_i(s) | \ s \notin S \cap Q_i(\theta_i) \}$.
First note that the agent's utility of this set is
$u_i(R_i) = \frac{{\hat A}_i + a_i \delta \delta'}{{\hat B}_i +b_i\delta}$,
because $S$ hashes into $(a,b)$.
Since $(a,b)$ is consistent with $(\theta,D)$,
we have $\theta'_i \leq u_i(R_i) < \theta_i$.
From the definition of  $Q_i(\theta_i)$ and $f_i(s)$, it follows
that $R_i$ is locally optimal, and hence globally optimal.

Since each Agent $i$ adopts the set $S \cap Q_i(\theta_i)$ of platforms it follows from the definitions that 
$\textnormal{value}_{(\theta,D)}(S(\theta,D)) = \textnormal{profit} ( S(\theta,D))$.
\end{proof}

It is easy to see that the analogous lemma to Lemma \ref{lem.multi2}
also holds with a similar proof.
Optimality of the algorithm follows.

{\bf Theorem \ref{thm:multi-Designer-PDP}}.
The extension of Algorithm \ref{alg:Designer-multi} 
to the competitive setting computes an optimal solution to
the Designer's problem. It runs in polynomial time for fixed number of
Agents, under the stated assumptions on the input parameters.

The algorithm can be extended to the case where the designer can choose between multiple possible platforms for each state.

%% file: appendixF.tex
\section{A PDP Game with No Pure Nash Equilibrium}
\label{appendixF}

What if the setting, as is often the case the modern internet economy, was more dynamic, and designers competed over platforms? In other words, each Designer is allowed to build a platform at each state; the Agent then chooses which platforms to accept. We call the simultaneous game the Designers play the PDP game.

We show that there are instances of the PDP game (even in the flower setting) where there are no pure Nash equilibria: 
\begin{lemma}
\label{lemma:no_pure_nash_instance}
There exist instances of the flower PDP game which have no pure Nash equilibria. 
\end{lemma}
\begin{proof}
Consider an instance of the game where there are $2$ Designers and $3$ states. In the following, the superscript will denote which designer is being referred to, while the subscript will denote the state. Let the Designers' rewards be
\[
d_1^1 = 100;\quad d_2^1 = 0;\quad d_3^1 = 50
\]
\[
d_1^2 = 0;\quad d_2^2 = 100;\quad d_3^2 = 2000
\]
Suppose $\cL{i} = 0$ for all $i \in [3]$ for the Agent. Let the Agent's rewards be 
\[
{\cP{1}}^1 = 50;\quad {\cP{2}}^2 = 50;\quad {\cP{3}}^1 = 2000;\quad {\cP{3}}^2 = 1000
\]
and for the rest $0$. 
Suppose that all $z_i, \lambda_i = 1$ always and all costs are $\eps = 0.001$. Then $A = 0, B = 3 + 1 = 4$. 

We now show there can be no pure Nash equilibrium by showing that 
for any pair of strategies for Designers $1$ and $2$, one of them always wants to deviate. We only need consider disjoint strategies (where the Designers never build at the same state -- this is always clearly sub-optimal for the Designer whose platform the Agent does not choose, since costs are positive).

We will denote strategies as tuples with Designer $1$ being the first entry. First note that both Designers should always build a platform somewhere (the cost is small enough for building a platform compared to the profit, and there is always at least one state for each Designer that they always win (state $1$ for Designer $1$, and state $2$ for Designer $2$)).
Now we note that strategy $(\{1\}, \{2\})$ dominates all strategies where Designer $2$ does not build at state $3$ (Designer $1$ always does better building only at state $1$). We also note that strategy $(\{1, 3\}, \{2\})$ dominates all strategies where Designer $2$ does build at state $3$ (for Designer $2$, since Designer $2$ never wins state $3$ and unnecessarily pays a positive cost), and $(\{1\}, \{3\})$ dominates $(\{1\}, \{2\})$ (for Designer $2$, since Designer $2$ always prefers to build at state $3$ over state $2$ if Designer $1$ does not build there). Thus from any choice of strategy pair, we arrive at the following cycle between strategies as the Designers constantly change their minds: 
\[
(\{1\}, \{2\}) \to (\{1\}, \{3\}) \to (\{1, 3\}, \{2\}) \to \cdots
\]
One can check this by simply noting that Designer $2$ always wants to build at state $3$ to improve their profit, but this has the ill effect of stopping the Agent from going to state $1$, which means Designer $1$ builds at state $3$. The Agent always picks Designer $1$ over Designer $2$, so Designer $2$ builds only at state $2$ again, but this means that Designer $1$ would prefer to only build at state $1$ since they get more profit from state $1$. And so the cycle continues, and there is no pure Nash equilibrium. 
\end{proof}